\documentclass[12pt]{article}
\usepackage[raggedright]{titlesec}
\usepackage[shortcuts]{extdash}
\usepackage{graphicx}
\usepackage{amsfonts, amsmath, amssymb, amsthm, mathrsfs, mathtools, amscd}
\usepackage{braket}
\usepackage{graphicx}
\usepackage[all]{xy}
\usepackage{color}

\usepackage{enumerate}
\usepackage{bm}
\usepackage{graphicx}
\usepackage{nicefrac}
\usepackage{color}
\usepackage[all]{xy}
\usepackage{hyperref}
\usepackage{bbm}
\usepackage{tikz}
\usepackage{tikz-cd}
\usepackage{cancel}
\usepackage{xcolor}
\usetikzlibrary{positioning}
\usepackage{MnSymbol}

\newtheorem{theorem}{Theorem}

\newtheorem{definition}{Definition}

\newtheorem{proposition}[theorem]{Proposition}

\newtheorem{lemma}[theorem]{Lemma}

\newtheorem{corollary}[theorem]{Corollary}

\newtheorem{example}{Example}

\usepackage[margin=1in]{geometry}
\allowdisplaybreaks

\usepackage{hyperref}

\newcommand\ld{\lambda}
			
\newcommand\Sig{\underline{\Sigma}}	
\newcommand\join{\vee}							
\newcommand\bjoin{\bigvee}					
\newcommand\meet{\wedge}						
\newcommand\bmeet{\bigwedge}				
\newcommand\cH{\mathcal{H}}					
\newcommand\PH{\mathcal{P}(\cH)}		

\newcommand\CSet{\mathbf{Set}}			

\newcommand\op{\operatorname{op}}		

\newcommand\tphi{\tilde\phi}

\newcommand\mc[1]{\mathcal{#1}}
\newcommand\on[1]{\operatorname{#1}}
\newcommand\BH{\mc{B}(\cH)}
\newcommand\sa{\on{sa}}
\newcommand\BHsa{\BH_{\sa}}
\newcommand\cN{\mc{N}}
\newcommand\cNsa{\cN_{\sa}}
\newcommand\bbR{\mathbb{R}}
\newcommand\bbC{\mathbb{C}}
\newcommand\PN{\mc{P}(\cN)}
\newcommand\PV{\mc{P}(V)}
\newcommand\cJ{\mc{J}}
\newcommand\VN{\mc{V}(\cN)}
\newcommand\VH{\mc{V}(\cH)}
\newcommand\tr{\operatorname{tr}}
\newcommand\Vsa{V_{\sa}}
\newcommand\PPi{\underline{\Pi}}
\newcommand\tPPi{\underline{\tilde\Pi}}
\newcommand\ga{\gamma}
\newcommand\PpH{\mc{P}_1(\cH)}
\newcommand\Aut[1]{\on{Aut}(#1)}
\newcommand\POne{\underline{1}}
\newcommand\PP{\underline{P}}
\newcommand\spec[1]{\on{sp}(#1)}

\newcommand\ra{\rightarrow}

\newcommand\lra{\longrightarrow}
\newcommand\lmt{\longmapsto}
\newcommand\Lra{\Leftrightarrow}
\newcommand\hra{\hookrightarrow}

\begin{document}
\title{\textbf{Contextuality and the fundamental theorems of quantum mechanics}}
\author{Andreas D\"oring\thanks{info@andreas-doering.net} \hspace{0.025cm} and Markus Frembs\thanks{markus.frembs13@imperial.ac.uk}\hspace{0.025cm} $^{1}$}
\date{\small{\textit{$^1$Department of Physics, Imperial College, Prince Consort Road, SW7 2AZ, United Kingdom}}}
\maketitle

\abstract{Contextuality is a key feature of quantum mechanics, as was first brought to light by Bohr \cite{Bohr98} and later realised more technically by Kochen and Specker \cite{KocSpe67}. Isham and Butterfield put contextuality at the heart of their topos-based formalism and gave a reformulation of the Kochen-Specker theorem in the language of presheaves in \cite{IshBut98}. Here, we broaden this perspective considerably (partly drawing on existing, but scattered results) and show that apart from the Kochen-Specker theorem, also Wigner's theorem, Gleason's theorem, and Bell's theorem relate fundamentally to contextuality. We provide reformulations of the theorems using the language of presheaves over contexts and give general versions valid for von Neumann algebras. This shows that a very substantial part of the structure of quantum theory is encoded by contextuality.

\tableofcontents
\section{Introduction}
\textbf{Structural theorems of quantum theory.} There is a small number of key theorems in the foundations of quantum theory which throw the differences between classical and quantum in sharp relief. The first and oldest of these is Wigner's theorem \cite{Wig31} from 1931, which shows that each transformation on the set of pure states of a quantum mechanical system that preserves transition probabilities is given by conjugation with a unitary or anti-unitary operator. Hence, the pure state space of a quantum system has a very specific structure.

In 1957, Gleason \cite{Gle57} proved that any assignment of probabilities to projection operators such that probabilities assigned to orthogonal projections add up must be given by a quantum state already. Since projections represent propositions about the values of physical quantities\footnote{We use `physical quantity' synonymously with `observable'.} by the spectral theorem, Gleason's result justified the use of the Born rule---30 years after its introduction---when calculating expectation values in quantum mechanics.

Bell \cite{Bell1964} proved in 1964 that in local hidden variable theories, there is an upper bound on correlations that can exist between spatially separated subsystems of a composite system. Quantum theory violates this upper bound and hence cannot be (or be replaced by) a local hidden variable theory. There are a number of fine interpretational points, but it is largely accepted today that the violation of Bell's inequality has been confirmed experimentally.

Finally, in 1967 Kochen and Specker \cite{KocSpe67} showed that under mild and natural conditions, it is mathematically impossible to assign values to all physical quantities simultaneously. Usually, this is phrased as saying that there are no non-contextual value assignments. Since in classical physics, states do assign values to all physical quantities at once, this is a rather strong result and a severe obstacle to any realist interpretation of the quantum formalism.

Each of these landmark theorems singles out a central aspect of quantum theory that distinguishes it from classical physics. The contents of the theorems are very distinct; each one concerns a different structural aspect of quantum theory. Yet, in this article we will show that in fact all these theorems have a common source, which is \emph{contextuality}. This strongly suggests that contextuality is at the heart of quantum theory and is largely responsible for the structural differences between classical and quantum.

\textbf{Contextuality and its mathematical formalisation.} Contextuality, which was introduced by Bohr \cite{Bohr98}, is a deep concept, and like many other deep concepts in physics, it has taken on a range of meanings and interpretations in the literature \cite{Griffiths2019}. This has led to a certain danger of talking vaguely and at cross-purposes. In order to avoid this, we define precisely what we mean by contextuality and we give a rather minimal and conservative definition that comes with little interpretational baggage.

We say that a physical system has \emph{physical contextuality} if it has some incompatible physical quantities, that is, quantities that cannot be measured simultaneously in an arbitrary state. For all we know today, physical contextuality in this sense is a characteristic feature of all quantum systems, but not of classical systems, and the restrictions on co-measurability are fundamental and not just due to a lack of experimenters' (or theoreticians') finesse.

Even in a physical system with physical contextuality, there are families of compatible, co-measurable physical quantities. (In an extreme case, such a family could consist of a single physical quantity, although this does not occur for quantum systems.) A family of compatible physical quantities is called a \emph{physical context}. We remark that
\begin{itemize}
	\item Our definition of physical contextuality does not refer to actual measurement setups, but we could, as is often done, identify a measurement setup with a physical context, i.e., with the family of physical quantities that is measured by the setup.

	\item Our definition does not refer to values measured and/or possessed by physical quantities, nor to probabilities. This is not necessary for our purposes, and it avoids many interpretational issues that often cloud the discussion. In particular, we can consider whether non-contextual assignments of values, probabilities etc. are possible in our kind of contextual theory.
	
	\item In order to be compatible according to our definition, two physical quantities must be co-measurable in \emph{all} states. Hence, two physical quantities that are co-measurable in some states can still be incompatible and lie in different contexts. We are not concerned with weak measurements, weak values etc.
	
\end{itemize}

The definition of physical contextuality is given in a somewhat intuitive manner, since no precise mathematical formalisation of `physical quantities' and `states' is provided so far. We remark---following Kochen and Specker \cite{KocSpe67} and especially Conway and Kochen \cite{ConKoc06}---that much less than the full Hilbert space formalism needs to be given: as long as the physical system under consideration (or some subsystem of it) has the physical quantity usually called `spin-$1$', which can be measured in different directions in space, the system is contextual, since measurements in different directions cannot be performed simultaneously. The mathematical formalism required to describe this situation is a modest part of projective geometry in $3$ dimensions.

Yet, following standard practice, we will assume the usual Hilbert space formalism in which the set of (bounded) physical quantities is mathematically represented by the set $\BHsa$ of bounded self-adjoint operators on the Hilbert space $\cH$ of the system. The self-adjoint operators form the real part of the complex, noncommutative algebra of bounded operators on the Hilbert space of the system. A context is mathematically formalised as a commutative subalgebra of this noncommutative algebra. This can be generalised to von Neumann algebras of physical quantities. We will quickly recall the necessary mathematical background in Sec.~\ref{Sec_MathPrelims}. Moreover, we will introduce the context category and some minimal background on presheaves in order to make this article largely self-contained.

\textbf{Structural theorems and contextuality.} We will show that each of the key theorems mentioned above has an \emph{equivalent} reformulation in the language of presheaves over contexts, thus showing the close and sometimes surprising connections between these theorems and contextuality. The prototype of such results is found in the work by Isham, Butterfield, and Hamilton \cite{IshBut98,HIB00} on the Kochen-Specker theorem. We will extend their results to the other fundamental theorems by Wigner, Gleason, and Bell and will provide a bigger, more coherent picture of the role of contextuality in the foundations of quantum mechanics.

In Sec.~\ref{Sec_Wigner}, Wigner's theorem is treated. The relevant presheaf is trivial, since Wigner's theorem is based on the mere order of contexts, as we will show. The Kochen-Specker theorem and its reformulation are presented in Sec.~\ref{Sec_KS}. Here, the so-called spectral presheaf plays the key role. It can be seen as a generalised state space for a quantum system, and the Kochen-Specker theorem is equivalent to the fact that this space has no points in a suitable sense. Gleason's theorem is treated in Sec.~\ref{Sec_Gleason}. Its reformulation is based on the so-called probabilistic presheaf, which does have points, i.e., global sections, and these correspond exactly with quantum states. Finally, we consider Bell's theorem and its relations to contextuality in Sec.~\ref{Sec_Bell}. The relevant presheaf is a bipartite (or multipartite) version of the probabilistic presheaf. This presheaf is based on a simple way of composition via contexts, yet it is rich enough to encode all quantum correlations (and not more). In fact, by adding a notion of time orientation in subsystems it singles out quantum states unambiguously. Sec.~\ref{Sec_Conclusion} concludes.

\section{Mathematical preliminaries}	\label{Sec_MathPrelims}
\subsection{Algebras of physical quantities}	\label{Subsec_Algebra}
Throughout, we will take the perspective of \emph{algebraic quantum mechanics}, that is, we emphasise the role of the physical quantities, or observables, and the algebra they form. This means no departure from standard textbook quantum mechanics, just a slightly different perspective that allows for substantial generalisations. We will assume that the physical quantities generate a von Neumann algebra (more details below; standard references are e.g. \cite{KadRin83,KadRin86,Tak79}). This allows to encode symmetries of the quantum system directly at the algebraic level by picking a von Neumann algebra that has a non-trivial commutant. Superselection rules can be modelled algebraically by algebras with non-trivial center.

\textbf{The algebra of bounded operators, weak operator topology and norm topology.} Let $\cH$ be the complex Hilbert space of the quantum system under consideration. $\cH$ may be finite- or infinite-dimensional. The set of all bounded linear operators on $\cH$ is denoted $\BH$. Linear operators can be added, multiplied (by composition), and multiplied by complex numbers; $\BH$ is an algebra over $\bbC$. Of course, multiplication of operators is not commutative in general.

If $\cH$ is of finite dimension $n$, then $\cH = \bbC^n$ and $\BH = M_n(\bbC)$, the algebra of all $n\times n$-matrices with complex entries. If $\cH$ is infinite-dimensional, then $\BH$ carries several interesting topologies (which all coincide in finite dimensions). We will only consider the weak operator topology and the norm topology. For details and many more results on topologies on $\BH$ and how they relate to each other, see e.g. \cite{Tak79}. 

Let $\langle\psi,\eta\rangle$ denote the inner product of $\psi,\eta\in\cH$. Let $(a_i)_{i\in I}$ be a family of bounded operators on $\cH$, labeled by a directed set $I$. The \emph{weak operator topology} on $\BH$ is the topology of pointwise weak convergence, i.e., $a_i\ra a$ weakly if $\langle\psi,a_i\eta\rangle\ra\langle\psi,a\eta\rangle$ for all $\psi,\eta\in\cH$. We will simply say \emph{weak topology} from now on.

Let $a\in\BH$. The \emph{norm of $a$} is defined as
\[
			||a|| := \inf\{c\in\bbR \mid ||a\psi||\leq c||\psi|| \text{ for all } \psi\in\cH\},
\]
where $||\psi||=\sqrt{\langle\psi,\psi\rangle}$ is the norm of $\psi$. The topology on $\BH$ defined by the operator norm is called the \emph{norm topology} (or \emph{uniform topology}).

\textbf{Physical quantities, self-adjoint operators and von Neumann algebras.} Let $\cH$ be the Hilbert space of the quantum system under consideration. The physical quantities of the quantum system are represented by the bounded \emph{self-adjoint operators} on $\cH$. The fact that we only consider \emph{bounded} operators is not a severe restriction, since in the operator-algebraic framework there are ways of dealing with unbounded operators (by affiliating them to the von Neumann algebra of physical quantities, see e.g. section 5.6 in vol. 1 of \cite{KadRin83,KadRin86}).

We assume as usual that the self-adjoint operators representing physical quantities form a \emph{real} vector space under addition, denoted $O_{\sa}$. Multiplication (composition) of self-adjoint operators is not commutative in general, $ab\neq ba$, and the product $ab$ is self-adjoint if and only if $a$ and $b$ commute. The complexification of the set of self-adjoint operators,
\begin{align*}
			O_{\sa}+iO_{\sa} = \{a_1+ia_2 \mid a_i\in O_{\sa}\}\; ,
\end{align*}
is a complex vector space that is closed under multiplication, moreover it is a complex \emph{algebra}. By construction, it is a subalgebra of $\BH$, the algebra of all bounded operators on the Hilbert space $\cH$. We will always assume that this complex algebra is closed in the weak topology, i.e., the algebra is a \emph{von Neumann algebra} (see e.g. \cite{KadRin83}).\footnote{For finite-dimensional $\cH$, weak closure is automatic.} We denote this algebra of physical quantities by $\cN$. The real vector space of self-adjoint operators in $\cN$ is denoted $\cNsa$.

\textbf{Projections and the spectral theorem.} Recall that a \emph{projection operator} is a self-adjoint operator $p$ such that $p^2=p$. The \emph{rank} of a projection $p$ is the dimension of the subspace that $p$ projects onto (which may be infinite). By the \emph{spectral theorem}, each self-adjoint operator $a$ is the norm limit of finite real linear combinations of projections, i.e., $a$ can be approximated by operators of the form $\sum_{i=1}^n A_i p_i$, where the $A_i$ are real numbers and the $p_i$ are mutually orthogonal projections, that is, $p_i p_j=\delta_{ij}p_i$. In finite dimensions, $a$ is exactly of the form $a=\sum_{i=1}^n A_i p_i$, where the projection $p_i$ projects onto the eigenspace of the eigenvalue $A_i$ of $a$. In infinite dimensions, $a$ may fail to have any eigenvalues, but it has a non-empty spectrum.

We had assumed that the algebra $\cN$ is weakly closed. This guarantees that $\cN$ is generated by the projection operators that it contains in the following sense: let $a$ be a self-adjoint operator in $\cN$. All the projection operators $p_i$ in the approximations to $a$ of the form $\sum_{i=1}^n A_i p_i$ can be chosen from $\cN$. Moreover, any operator $b$ in a von Neumann algebra $\cN$ can be decomposed uniquely as $b=a_1+ia_2$, where $a_1, a_2$ are self-adjoint operators in $\cN$. Since both $a_1$ and $a_2$ can be approximated by real linear combinations of projections in $\cN$, also $b$ can be approximated by projections in $\cN$.

\textbf{The double commutant and weak closure.} There is another, more algebraic way of expressing weak closure. Let $S$ be a family of bounded operators on $\cH$. The \emph{commutant of $S$}, denoted $S'$, is the set of all operators in $\BH$ that commute with all operators in $S$,
\begin{align*}
			S'=\{b\in\BH \mid \forall a\in S:[a,b]=0\}\; .
\end{align*}
Von Neumann's \emph{double commutant theorem} (see e.g. \cite{KadRin83}) shows that a subalgebra $\cN$ of $\BH$ is weakly closed, i.e., a von Neumann algebra if and only if $\cN=\cN''$. Here $\cN''=(\cN')'$ is the commutant of the commutant of $\cN$. Moreover, if $\PN$ denotes the projections in a von Neumann algebra $\cN$, then one can show $\PN''=\cN$.

\begin{example}
The canonical example of a von Neumann algebra on a given Hilbert space is $\BH$, the algebra of \emph{all} bounded operators on $\cN$. In fact, this is the algebra (implicitly) used in most standard textbooks on quantum mechanics.

Clearly, $\BH$ is the largest von Neumann algebra on the given Hilbert space $\cH$. Any other von Neumann algebra $\cN$ formed by bounded operators operators on $\cH$ is a weakly closed subalgebra of $\BH$.

If $\cH$ is finite-dimensional, $\cH=\bbC^n$, the algebra $\BH$ simply is the algebra of all $n\times n$-matrices. Self-adjoint operators are given by Hermitian matrices. In finite dimensions, every operator is bounded and every subalgebra of $\BH$ is weakly closed.
\end{example}

\textbf{The lattice of projections.} The projections in a von Neumann algebra $\cN$ form a \emph{complete orthomodular lattice}, denoted $\PN$. The lattice operations can most easily be understood in geometric terms: let $S_p$ denote the closed subspace of Hilbert space $\cH$ onto which $p$ projects. (Clearly, there is a one-to-one correspondence between projection operators and closed subspaces of $\cH$.) The \emph{meet $p\meet q$} of two projections is the projection onto the subspace $S_p\cap S_q$, and the \emph{join $p\join q$} is the projection onto the closure of the subspace generated by $S_p$ and $S_q$. The largest projection is $1$, the identity operator, which projects onto the whole of $\cH$. The smallest projection is $0$, the zero projection, which projects onto the null subspace. Note that $\PH$ is \emph{not distributive}, that is, in general we have
\begin{align*}
			p\meet(q\join r) \neq (p\meet q)\join(p\meet r)\; .
\end{align*}
It is easy to find examples in any Hilbert space of dimension $2$ or greater: just take $p,q,r$ to be projections onto one-dimensional subspaces lying in a plane. 

Every family of projections $(p_i)_{i\in I}$ has a least upper bound (join) in $\PN$, denoted $\bjoin_{i\in I}p_i$, and also has a greatest lower bound (meet) in $\PN$, denoted $\bmeet_{i\in I}p_i$, so $\PN$ is a \emph{complete} lattice. Moreover, there is a \emph{complement} defined on $\PN$, given by $1-p$ for $p\in\PH$. Geometrically, $1-p$ is the projection onto the orthogonal complement of the subspace $S_p$ that $p$ projects onto. 

In quantum logic, the projection operators are interpreted as propositions about the values of physical quantities. Let $a$ be a self-adjoint operator representing some physical quantity. For simplicity, assume that $a=\sum_{i=1}^n A_i p_i$. Then the projection $p_i$ represents the proposition ``the physical quantity (represented by) $a$ has the value $A_i$''. The lattice operations $\meet,\join$ and the complement are interpreted logically as And, Or, and Not. Non-distributivity and other conceptual problems make it difficult to interpret quantum logic in any straightforward way, though.


\textbf{States on a von Neumann algebra.} A \emph{positive operator} $a$ in a von Neumann algebra is an operator that is of the form $a=b^2$ for some self-adjoint operator $b\in\cNsa$. Clearly, positive operators are self-adjoint. A \emph{state} on a von Neumann algebra $\cN$ is a linear map $\rho:\cN\ra\bbC$ such that
\begin{itemize}
	\item [(a)] $\rho(1)=1$,
	\item [(b)] for all positive operators $a\in\cNsa$, $\rho(a)\geq 0$.
\end{itemize}
Hence, a state is a positive linear functional of norm $1$. The states on a von Neumann algebra $\cN$ form a convex set $\mc S(\cN)$. The extreme points of $\mc S(\cN)$ are called \emph{pure states}. In the case of $\cN=\BH$, the pure states are the familiar vector states.

A state $\rho$ is called \emph{multiplicative} if $\rho(ab)=\rho(a)\rho(b)$ for all $a,b\in\cNsa$. This is a strong condition: multiplicative states exist only on commutative von Neumann algebras. For these, they are exactly the pure states.

A state $\rho$ is called \emph{normal} if $\rho(a_i) \rightarrow \rho(a)$ for every monotone increasing net of operators $a_i \in \cN$ with least upper bound $a$. Equivalently, $\rho$ is normal if it is completely additive, $\rho(\sum_{i\in I}p_i) = \sum_{i \in I} \rho(p_i)$, for every orthogonal family $(p_i)_{i\in I}$ of projections in $\PN$. Normal states correspond to density matrices (or, more precisely, to positive operators which have trace $1$), that is, every normal state $\rho:\cN\ra\bbC$ is of the form
\begin{align*}
			\forall a\in\cN: \rho(a)=\tr(\varrho a)\; ,
\end{align*}
where $\varrho$ is a density matrix (in finite dimensions), or more generally a positive operator of trace $1$. We will often use $\rho$ for both the state and its corresponding positive operator. Note that in finite dimensions, every self-adjoint operator has a finite trace, but in infinite dimensions, having a finite trace is a proper condition. Those operators that have a finite trace are called \emph{trace-class}.

Just as in standard quantum mechanics, in algebraic quantum theory (mathematical) states on a von Neumann algebra are interpreted as physical states of the quantum system, assigning expectation values to physical quantities.

\textbf{Jordan algebras and associativity.} Instead of the usual multiplication (by composition) of self-adjoint operators, one can also use another product, defined by
\begin{align*}
			\forall a,b\in\BHsa: a\cdot b := \frac{1}{2}(ab+ba)=\frac{1}{2}\{a,b\}\; .
\end{align*}
This is the \emph{Jordan product}, given by the anti-commutator of $a$ and $b$ up to the conventional factor $\frac{1}{2}$. In contrast to the usual product $ab$, the Jordan product $a\cdot b$ is always self-adjoint, even if $a$ and $b$ do not commute. Hence, there is a \emph{real Jordan algebra} $(\BHsa,\cdot)$. Its complexification is $(\BH,\cdot)$.

Clearly, the Jordan product is commutative, $a\cdot b=b \cdot a$ for all $a,b\in\BHsa$. To every von Neumann algebra $\cN$, we can associate a weakly closed Jordan algebra $\cJ(\cN)$ by replacing the generally noncommutative product in $\cN$ by the commutative Jordan product. There is a `shadow' of noncommutativity left: it is easy to show that $\cJ(\cN)$ is associative if and only if $\cN$ is commutative.

\subsection{Contexts and the context category}	\label{Subsec_Contexts}
\textbf{Mathematical contexts.} We begin with the basic definitions of a (mathematical) context and the partially ordered set of all contexts of a quantum system.
\begin{definition}
Let $\cH$ be the Hilbert space of the quantum system under consideration and let $\cN\subseteq\BH$ denote the von Neumann algebra of physical quantities of the system. A \emph{context} is a commutative von Neumann subalgebra of $\cN$.
\end{definition}
This definition of a (mathematical) context simply encodes the idea that within a (physical) context, which can be identified with a chosen experimental setup, all the physical quantities are compatible, co-measurable, and hence are represented mathematically by commuting self-adjoint operators. We denote contexts as $V,\tilde{V},\dot{V},...$

\textbf{The context category.} There is a natural partial order on contexts: some contexts are maximal, that is, one cannot add any further self-adjoint operators to them without destroying commutativity. Other contexts are non-maximal, they are commutative von Neumann subalgebras that are properly contained in a larger context. In fact, if $V$ is a non-maximal context, there are many different maximal contexts containing $V$. This can already be seen in finite dimensions:

\begin{example}
Consider a spin-$1$ system with Hilbert space $\cH=\bbC^3$ and algebra of physical quantities $\BH = M_3(\bbC)$. Recall that the physical quantities of the system are represented by self-adjoint operators in $\BH$. Let $p_1$ be a projection onto a one-dimensional subspace, and let
\begin{align*}
			V_1:=\{p_1,1\}''=\bbC p_1+\bbC 1
\end{align*}
be the non-maximal context generated by $p_1$. The projections in $V$ are $0,p_1,1-p_1$ and $1$. Let $p_2,p_3$ be two mutually orthogonal rank-$1$ projections such that $p_2+p_3=1-p_1$ (that is, $p_2$ and $p_3$ are also orthogonal to $p_1$). Then the context
\begin{align*}
			V_{123}:=\{p_1,p_2,p_3\}''=\bbC p_1+\bbC p_2+\bbC p_3
\end{align*}
is maximal and contains $V_1$. In fact, the context $V_1$ is contained in continuously many maximal contexts: there are infinitely many pairs of mutually orthogonal rank-$1$ projections $p'_2,p'_3$ such that $p'_2+p'_3=1-p_1$, and each set $\{p_1,p'_2,p'_3\}$ generates a maximal context $V_{12'3'}=\{p_1,p'_2,p'_3\}''$ that contains $V_1$. Geometrically, one rotates in the plane that $1-p_1$ projects onto to obtain continuously many pairs $p'_2,p'_3$.
\end{example}

Hence, any quantum system has many different contexts, with smaller ones contained in larger ones. The key idea for all that follows is that one should not just consider a single context of a quantum system, or a small number, but \emph{all} of them simultaneously. This idea goes back to Chris Isham \cite{IshBut98} and has become a fruitful perspective and powerful tool over the last 20 years. The topos approach to quantum theory and many subsequent developments are based on this idea, for an introduction see e.g. \cite{DoeIsh11}. The key definition is:

\begin{definition}
Let $\cH$ be the Hilbert space of the quantum system under consideration and let $\cN\subseteq\BH$ denote the von Neumann algebra of physical quantities. The \emph{context category} of the system is the set of all contexts, i.e., the set of all commutative von Neumann subalgebras of $\cN$, equipped with inclusion as partial order. The context category is denoted $\VN$. If $\cN=\BH$, then we will simply write $\VH$ (instead of $\mc{V}(\BH)$).
\end{definition}

The name \emph{context category} comes from the fact that every partially ordered set (or poset, for short) can also be regarded as a category \cite{McL71}. The objects are the elements of the poset, and there is an arrow $a\ra b$ if and only if $a\leq b$. Hence, in a poset seen as a category, arrows express the order, so there is at most one arrow between any two objects. We will actually need very little category theory in the following, but we will feel free to use some simple and well-established categorical notions where appropriate. The reader not familiar with category theory can equally well read `context poset' for `context category' throughout. If a context $\tilde{V}$ is contained in another context $V$, we will write $i_{\tilde{V} V}:\tilde{V}\hra V$ for the inclusion map. Alternatively, one could just write $\tilde{V}\subset V$.

The following, powerful result is due to Harding and Navara \cite{HarNav11}:
\begin{theorem}
Let $\cN$ be a von Neumann algebra not isomorphic to $\bbC\oplus\bbC$ or to $M_2(\bbC)$. Then the context category $\VN$ of $\cN$ determines the projection lattice $\PN$ as an orthomodular lattice up to isomorphism. Conversely, the projection lattice $\PN$ determines the poset $\VN$ up to isomorphism.
\end{theorem}
In fact, Harding and Navara's proof holds more generally for orthomodular lattices with no maximal Boolean sublattices with only $4$ elements (this is why we exclude the trivial cases $\cN=\bbC\oplus\bbC$ and $\cN=M_2(\bbC)$). The result shows that the context category, i.e., the set of contexts together with the information of how contexts are contained within each other, encodes exactly the same amount of information as the projection lattice. In this sense,
\begin{center}
			\textbf{Contextuality determines quantum logic and vice versa.}
\end{center}
Note that the context category $\VN$ is just a poset. Its elements, the contexts $V\subseteq\cN$, are just `points' within $\VN$, without inner structure. In particular, from the perspective of $\VN$, we do not have access to the commutative von Neumann subalgebras, much less to the operators contained within each context. All the structure of $\VN$ lies in the order, that is, in the information of how some contexts are contained within others. This makes the Harding-Navara result quite remarkable, since the mere order structure on contexts determines the full structure of the projection lattice.

\textbf{Contexts without (non)commutativity.} The context category $\VN$ of a von Neumann algebra $\cN$ can be defined without any reference to (non)commutativity: every weakly closed associative Jordan subalgebra of $\cN$ is a commutative von Neumann subalgebra and vice versa. Hence, we can regard $\cN$ as a weakly closed Jordan algebra and consider the set of its weakly closed associative Jordan subalgebras, partially ordered by inclusion. This poset is (isomorphic to) $\VN$.

\subsection{Presheaves over the context category}
\textbf{The concept of a presheaf: local data glued together.} We saw in the previous subsection that the context category $\VN$ already encodes a lot of information about a quantum system. Now we will build further structures upon the context category in order to make it an even more useful tool. Concretely, given the context category $\VN$, we are interested in assigning data to each context. Moreover, since $\VN$ is a poset, we want to relate the data assigned to a context $V$ to the data assigned to a smaller context $\tilde{V}\subset V$.

For example, for each context $V\in\VN$, one may consider the set $\Sigma(V)$ of all pure states of $V$. If $\tilde{V}\subset V$, then every pure state of $V$ gives a pure state of $\tilde{V}$ simply by restriction,
\begin{align*}
			\Sigma(V) &\lra \Sigma(\tilde{V})\\
			\ld &\lmt \ld|_{\tilde{V}}\; .
\end{align*}
(Pure states of commutative von Neumann algebras are traditionally denoted $\ld$.) In this way, we obtain a natural map from the pure states of $V$ to the pure states of $\tilde{V}$, hence, relating the data assigned to $V$ to the data assigned to $\tilde{V}$.

This is an example of a general construction, viz. a \emph{presheaf}. This naming is traditional and has no particular meaning for us. The general definition of a presheaf over the context category $\VN$ is:
\begin{definition}
Let $\cN$ be a von Neumann algebra and let $\VN$ be its context category. A \emph{presheaf over $\VN$} is a contravariant functor $\PP:\VN\ra\CSet$. That is, $\PP$ is given
\begin{itemize}
	\item [(i)] on objects: for all $V\in\VN$, $\PP_V$, \emph{the component of $\PP$ at $V$}, is some set,
	\item [(ii)] on arrows: for all inclusions $i_{\tilde{V}V}:\tilde{V}\hra V$, there is a \emph{restriction map}
	\begin{align*}
				\PP(i_{\tilde{V}V}): \PP_V &\lra \PP_{\tilde{V}}\\
				x &\lmt \PP(i_{\tilde{V}V})(x)\; .
	\end{align*}
\end{itemize}
\end{definition}
Of course, this is a very general notion. The idea is that the `local' data assigned to each context can vary from context to context, but there are `connecting maps' relating the local data at $V$ and $\tilde{V}$ whenever $\tilde{V}\subset V$. The mathematical language of presheaves is a convenient tool for book-keeping.

\begin{example}\textbf{\emph{The trivial presheaf.}} The simplest (non-empty) presheaf over $\VN$ is the \emph{trivial presheaf} $\POne$, which is given
\begin{itemize}
	\item [(i)] on objects: for all $V\in\VN$, $\POne_V:=\{*\}$, the one-element set,
	\item [(ii)] on arrows: for all inclusions $i_{\tilde{V}V}:\tilde{V}\hra V$, 
	\begin{align*}
				\POne(i_{\tilde{V}V}): \POne_V &\lra \POne_{\tilde{V}}\\
				* &\lmt *\; .
	\end{align*}
\end{itemize}
\end{example}

We will make use of the trivial presheaf when we consider Wigner's theorem in Sec.~\ref{Sec_Wigner}. For the treatment of the Kochen-Specker theorem in Sec.~\ref{Sec_KS}, we will use the so-called \emph{spectral presheaf} (already sketched above), for Gleason's theorem (in Sec.~\ref{Sec_Gleason}), we will use the so-called \emph{probabilistic presheaf}, and for Bell's theorem (in Sec.~\ref{Sec_Bell}), the so-called \emph{Bell presheaf}, a version of the probabilistic presheaf adapted to bipartite (or multipartite) systems. In fact, each of these presheaves is tailor-made for reformulating the respective theorem. The spectral presheaf is built from pure states in each context, and the probabilistic presheaf is built from mixed states.

\textbf{Contravariance and coarse-graining.} One may wonder why we are using \emph{contra}variant functors rather than covariant ones. If $V,\tilde{V}$ are contexts such that $\tilde{V}\subset V$, why not map the local data assigned to $\tilde{V}$ into the local data assigned to $V$? Covariant functors do have their place in the bigger scheme \cite{HLS09}, but, as it turns out, for our purposes we only need contravariant functors. Generically, the idea is that the data assigned to a larger context $V$ is richer, more informative, and can be coarse-grained or restricted to the data assigned to a smaller context $\tilde{V}$. Conversely, there is often no canonical way to `fine-grain' or extend data assigned to $\tilde{V}$ to data assigned to $V$. It is always possible to discard information, but it is often impossible to create information, at least not in a unique way. E.g. every pure state on $V$ gives a pure state on $\tilde{V}$ by restriction, but a pure state on $\tilde{V}$ can usually be extended in many different ways to a pure state of $V$. (The problem here is that there is no \emph{canonical} way of extending.)

\textbf{Mapping a presheaf into itself.} In order to relate Wigner's theorem to the trivial presheaf in Sec.~\ref{Sec_Wigner}, we have to consider automorphisms of the trivial presheaf, that is, reversible mappings of $\POne$ into itself. There are a number of possible definitions and conventions, but we will focus on a very simple notion of automorphism that suits our purposes.

Let $\PP$ be a presheaf over the context category $\VN$. Roughly speaking, we can map $\PP$ to itself by first shifting the components around and then mapping each (shifted) component into itself in a way that is compatible with the restriction maps (natural transformation). More precisely, the shifting around of components is achieved by a morphism of the base category, which is $\VN$ in our case, acting by pullback. This means that if $\tphi:\VN\ra\VN$ is a morphism of the base category, then it acts as follows: $\PP\circ\tphi$ is the presheaf over $\VN$ with component $(\PP\circ\tphi)_V=\PP_{\tphi(V)}$ at $V$. That is, the new component at $V$ is the old component at $\tphi(V)$, for every $V\in\VN$. The restriction maps of $\PP\circ\tphi$ are given in the obvious way by $(\PP\circ\tphi)(i_{\tilde{V}V})=\PP(i_{\tphi(\tilde{V})\tphi(V)})$.

An automorphism $\Theta:\PP\ra\PP$ of a presheaf $\PP$ over $\VN$ then consists of
\begin{itemize}
	\item [(a)] an automorphism $\tphi:\VN\ra\VN$ of the base category acting by pullback, thus mapping $\PP$ to $\PP\circ\tphi$,
	\item [(b)] followed by, for each $V\in\VN$, an isomorphism $\theta_V:(\PP\circ\tphi)_V\ra(\PP\circ\tphi)_V$ such that, whenever $\tilde{V}\subset V$, one has $\PP(i_{\tphi(\tilde{V})\tphi(V)})\circ\theta_V=\theta_{\tilde{V}}\circ\PP(i_{\tphi(\tilde{V})\tphi(V)})$.
\end{itemize}

\textbf{Local and global sections of a presheaf.} Presheaves over the context category $\VH$ (or $\VN$) are not just sets, but collections of sets (one for each context), which are interconnected by the restriction maps. Hence, the notion of an `element' of a presheaf must be defined suitably. Let $\PP$ be a presheaf over $\VN$ and let $D$ be a downward closed subset of $\VN$, i.e., if $V\in D$ and $\tilde{V}\subset V$, then $\tilde{V}\in D$. A \emph{local section $\gamma$ of $\PP$ over $D$} consists of a choice of one element from the component $\PP_V$ for each $V\in D$, denoted $\gamma_V$, such that, whenever $V,\tilde{V}\in D$ with $\tilde{V}\subset V$, then $\PP(i_{\tilde{V}V})(\gamma_V)=\gamma_{\tilde{V}}$. This condition means that the elements $\gamma_V$ that we pick from the sets $\PP_V$ (where $V\in D$) fit together under the restriction maps of the presheaf $\PP$. 

One should think of a local section $\gamma$ of $\PP$ over $D$ as a `partial' element of $\PP$. If one has a local section over $D=\VN$, then $\gamma$ is called a \emph{global section} (or \emph{global element}) of the presheaf $\PP$. This is the analogue of an element of a set or a point of a space.

For a given presheaf $\PP$, global sections may or may not exist (while local sections always exist, just make $D$ small enough). In fact, finding a global section amounts to fitting specified local data, one element from each component of $\PP$, together into a whole. We will see that the presheaf reformulations of the Kochen-Specker theorem (following Isham, Butterfield, and Hamilton), Gleason's theorem, and also Bell's theorem are statements about the existence or nonexistence of global sections of certain presheaves.

\section{Wigner's theorem, contextuality, and Jordan algebra structure}	\label{Sec_Wigner}

\subsection{Wigner's theorem, Dye's theorem, and Jordan $*$-automorphisms}
We first consider Wigner's theorem \cite{Wig31}:
\begin{theorem}
    (Wigner's theorem.) Let $\cH$ be a Hilbert space, $\mathrm{dim}(\cH) \geq 2$, and let $\PpH$ be the set of rank-$1$ projections on $\cH$ (equivalently, $\PpH$ is the projective Hilbert space). Every bijective map
    \begin{align*}
    			\varphi:\PpH &\lra \PpH\\
    			p &\lmt \varphi(p)
    \end{align*}
    such that $\mathrm{tr}(\varphi(p)\varphi(q)) = \mathrm{tr}(pq)$
    for all $p,q\in\PpH$ (i.e., transition probabilities are preserved) is implemented by conjugation with a unitary or anti-unitary operator $u$,
    \begin{align*}
    			\forall p\in\PpH: \varphi(p)=upu^*\; .
    \end{align*}
\end{theorem}
Various nice proofs can be found in the literature, for a modern perspective see e.g. Cassinelli et al. \cite{CLL04}. These authors also prove the following: let $\Aut{\PpH}$ denote the group of automorphisms of $\PpH$ (i.e., bijective maps $\varphi:\PpH\ra\PpH$ that preserve transition probabilities). If $\dim(\cH)\geq 3$, then $\Aut{\PpH}$ is isomorphic to the group $\Aut{\PH}$ of automorphisms of the projection lattice $\PH$, i.e., maps
\begin{align*}
			\phi:\PH &\lra \PH\\
			p &\lmt \phi(p)
\end{align*}
that are 
\begin{enumerate}
	\item bijective,
	\item preserve complements, $\forall p\in\PH:\phi(1-p)=1-\phi(p)$, and
	\item preserve and reflect order, $\forall p,q\in\PH: (p\leq q) \Lra (\phi(p)\leq\phi(q))$.
\end{enumerate}
Geometrically, $p\leq q$ means that the closed subspace that $p$ projects onto is contained in the closed subspace that $q$ projects onto. Algebraically, $(p\leq q)\Lra(pq=p)$. Since an automorphism $\phi\in\Aut{\PH}$ preserves the order, it also preserves all meets (greatest lower bounds) and all joins (least upper bounds) in $\PH$. Since $\phi$ also preserves complements, it is an automorphism of the complete orthomodular lattice $\PH$.

Under the group isomorphism $\Aut{\PpH}\ra\Aut{\PH}$, the automorphism $\phi\in\Aut{\PH}$ corresponding to a given $\varphi\in\Aut{\PpH}$ is an extension of $\varphi$ to all projections that preserves joins and complements (and hence also meets; for details, see \cite{CLL04}). Of course, we have $\phi|_{\PpH}=\varphi$.

Hence, if the Hilbert space is at least three-dimensional, Wigner's theorem is equivalent to the fact that every automorphism of the projection lattice $\PH$ is implemented by conjugation with a unitary or anti-unitary operator,
\begin{align*}
			\forall p\in\PH: \phi(p)=upu^*\; .
\end{align*}

There is a generalisation of Wigner's theorem to von Neumann algebras, which is closer to the formulation with automorphisms of $\PH$ than automorphisms of $\PpH$. This is Dye's theorem \cite{Dye55}. Before we formulate the theorem, we recall that given a von Neumann algebra $\cN$, we can form the associated Jordan algebra $(\cN,\cdot)$, which has the same elements and linear structure as $\cN$, and Jordan product given by
\begin{align*}
			\forall a,b\in\cN: a\cdot b:=\frac{1}{2}(ab+ba)=\frac{1}{2}\{a,b\}\; .
\end{align*}
Also recall that the Jordan product is commutative, but only associative if the von Neumann algebra $\cN$ is commutative. A Jordan $*$-automorphism of $(\cN,\cdot)$ is a bijective map $\Phi:(\cN,\cdot)\ra(\cN,\cdot)$ such that both $\Phi$ and $\Phi^{-1}$ preserve the Jordan product and the involution $(\_)^*$, that is,
\begin{align*}
			&\forall a,b\in(\cN,\cdot): \Phi(a\cdot b)=\Phi(a)\cdot\Phi(b),\\
			&\forall a\in(\cN,\cdot): \Phi(a^*)=\Phi(a)^*\; ,
\end{align*}
and analogously for $\Phi^{-1}$. We can now formulate Dye's theorem:
\begin{theorem}	\label{Thm_Dye}
Let $\cN$ be a von Neumann algebra with no direct summand of type $I_2$. For every automorphism $\phi:\PN\ra\PN$ of the projection lattice of $\cN$, there exists a unique Jordan $*$-automorphism $\Phi:(\cN,\cdot)\ra(\cN,\cdot)$ such that $\Phi(p)=\phi(p)$ for all projections $p\in\PN$.
\end{theorem}

It is easy to see that the Jordan $*$-automorphism $\Phi$ induced by an automorphism $\phi:\PN\ra\PN$ is \emph{ultraweakly continuous} (or normal), i.e., it preserves (countable) joins of projections \cite{Doe14}. Conversely, every ultraweakly continuous Jordan $*$-automorphism $\Phi:(\cN,\cdot)\ra(\cN,\cdot)$ induces an automorphism $\phi$ of the complete orthomodular lattice $\PN$ by $\phi:=\Phi|_{\PN}$.

The ultraweakly continuous Jordan $*$-automorphisms of $(\cN,\cdot)$ form a group denoted $\Aut{\cN,\cdot}$. Dye's theorem hence shows that, provided $\cN$ has no type $I_2$ summand, there is a group isomorphism
\begin{align*}
			\Aut{\PN}\lra\Aut{\cN,\cdot}
\end{align*}
between the group of automorphisms of the projection lattice and the group of ultraweakly continuous Jordan $*$-automorphisms of $\cN$.

One may wonder how Dye's theorem and Jordan $*$-automorphisms relate to unitary and anti-unitary operators as in Wigner's theorem. To see this, we first need the following well-known result (see e.g. \cite{AlfShu01}):
\begin{proposition}	\label{Prop_JordDecomp}
Every Jordan $*$-automorphism $\Phi:\cN\ra\cN$ of a von Neumann algebra $\cN$ can be decomposed as the sum of a $*$-isomorphism and a $*$-anti-isomorphism. 
\end{proposition}

More concretely, there are projections $p,q$ in the center of $\cN$ such that $\cN$ is unitarily equivalent to both $\cN p\oplus\cN (1-p)$ and $\cN q\oplus\cN (1-q)$, and $\Phi|_{\cN p}:\cN p\ra\cN q$ is a $*$-isomorphism, while $\Phi|_{\cN (1-p)}:\cN (1-p)\ra\cN (1-q)$ is a $*$-anti-isomorphism. Moreover, we need (see e.g. \cite{AlfShu01}):
\begin{proposition} \label{Prop_BHJord}
Every $*$-automorphism $\Phi:\BH\ra\BH$ is implemented by conjugation with a unitary operator and every $*$-anti-automorphism is implemented by conjugation with an anti-unitary operator.
\end{proposition}

By Dye's theorem and Prop.~\ref{Prop_JordDecomp}, every (ultraweakly continuous) Jordan $*$\-/automorphism $\Phi:\BH\ra\BH$ decomposes into a $*$-isomorphism on $\BH p$ and a $*$-anti-isomorphism on $\BH (1-p)$, where $p$ is a \emph{central} projection in $\PH$. Since $\BH$ is a factor (a von Neumann algebra with trivial center), the only central projections are $0$ and $1$. Thus, a Jordan $*$-automorphism $\Phi:\BH\ra\BH$ is either a $*$-automorphism or a $*$-anti-automorphism. Hence, by Prop.~\ref{Prop_BHJord} it is of the form
\[
			\Phi(a) = uau^*
\]
for a unitary or anti-unitary operator $u$ acting on $\BH$. We see that Wigner's theorem is a special case of Dye's theorem (depending on some special features of $\BH$) and we can rephrase Wigner's theorem as follows:

\begin{theorem}
    (Wigner's theorem in `Jordan formulation'.) Let $\cH$ be a Hilbert space, $\dim\cH\geq 3$. Every automorphism $\phi:\PH\ra\PH$ is implemented by a unique ultraweakly continuous Jordan $*$-automorphism $\Phi:(\BH,\cdot)\ra(\BH,\cdot)$ such that $\phi(p)=\Phi(p)$ for all $p\in\PH$.
\end{theorem}
This shows that contrary to the usual hand-waving arguments, there is a good mathematical reason to consider both unitary and anti-unitary operators in Wigner's theorem, since we actually have a statement about the structure of $\BH$ as a Jordan algebra. The Jordan structure is preserved by the action of both unitary and anti-unitary operators.

\subsection{The D\"oring-Harding result, contextuality and Wigner's theorem}
So far, all this does not relate to contexts in any obvious way. The result that connects Wigner's theorem (and Dye's theorem) with contextuality is the following:
\begin{theorem}	\label{Thm_DoeHar}
    (D\"oring, Harding \cite{DoeHar16}) Let $\cN$ be a von Neumann algebra not isomorphic to $\bbC\oplus\bbC$ or to $M_2(\bbC)$. For every order automorphism $\tphi:\VN\ra\VN$ of the context category of $\cN$, there is a unique (ultraweakly continuous) Jordan $*$-automorphism $\Phi:(\cN,\cdot)\ra(\cN,\cdot)$ such that $\tphi(V)=\Phi[V]$ for all $V\in\VN$.
\end{theorem}
This shows that the mere order structure of contexts determines the algebra of observables as a Jordan algebra up to isomorphism. The proof proceeds in two steps: first, using the result by Harding and Navara \cite{HarNav11} already mentioned in Sec.~\ref{Subsec_Contexts}, one shows that an order automorphism $\tphi:\VN\ra\VN$ induces a unique automorphism $\phi:\PN\ra\PN$ of the projection lattice, second, by Dye's theorem, this gives a Jordan $*$-automorphism $\Phi:(\cN,\cdot)\ra(\cN,\cdot)$. As a shorthand,
\begin{center}
			\textbf{Contextuality determines the Jordan algebra}\\
			\textbf{of physical quantities and vice versa.}
\end{center}
As remarked in \cite{Doe14}, it is easy to see that the Jordan $*$-automorphism $\Phi$ induced by an order isomorphism $\tphi:\VN\ra\VN$ is ultraweakly continuous. Hence, there is a group isomorphism
\begin{align*}
				\Aut{\VN} \lra \Aut{\cN,\cdot}\; .
\end{align*}

The D\"oring-Harding result can be regarded as a reformulation of Dye's theorem (Thm. \ref{Thm_Dye}), explicitly showing that it is contextuality which determines Jordan algebra structure of von Neumann algebras (and vice versa). Specialising to the algebra $\cN=\BH$ and using Prop.~\ref{Prop_BHJord}, we have:
\begin{theorem}	\label{Thm_WignerContextual}
(Wigner's theorem in contextual form.) Let $\cH$ be a Hilbert space, $\mathrm{dim}(\cH) \geq 3$. For every order automorphism $\tphi:\VH\ra\VH$, there is a unique unitary or anti-unitary operator $u$ such that
\begin{align*}
			\forall V\in\VH: \tphi(V) = uVu^*\; .
\end{align*}
Conversely, every unitary or anti-unitary operator $u$ induces an order automorphism of the context category $\VH$ by conjugation.
\end{theorem} 
This is our first reformulation of Wigner's theorem. Remarkably, any bijective map $\phi$ that preserves the order on the collection of contexts must be implemented by a unitary or anti-unitary operator. 

\subsection{Wigner's theorem and the trivial presheaf over the context category}
Obviously, the structure of the context category $\VH$ is sufficient to reformulate Wigner's theorem. Extra information, as may be provided by presheaves over $\VH$, is not necessary. In this sense, Wigner's theorem is the simplest of the theorems that we consider.

Yet, there is a reformulation of Wigner's theorem, more generally Dye's theorem, that does use a presheaf. Since we need exactly the information provided by the context category $\VN$, which is a partially ordered set, the presheaf must mirror this partial order (and nothing more). The trivial presheaf $\POne$ over $\VN$ does this: the component at $V\in\VN$ is the one-element set $\{*\}$, and for every inclusion $i_{\tilde{V}V}:\tilde{V}\hra V$, there is a restriction $\POne(i_{\tilde{V}V}):\POne_V\ra\POne_{\tilde{V}}$, sending $\{*\}$ to $\{*\}$. Hence, we both have the elements of the poset $\VN$ and the order relation encoded by $\POne$.\footnote{Strictly speaking, because $\POne$ is a contravariant functor, we encode the opposite order, hence the poset $\VN^{\op}$, but this gives exactly the same amount of information as $\VN$.}

As discussed in Sec.~\ref{Subsec_Contexts}, an automorphism $\Theta$ of a presheaf $\PP$ consists of two things (according to our convention): a shifting around of components, induced by an automorphism $\tphi$ of the base category acting by pullback, followed by an isomorphism $\theta_V:(\PP\circ\tphi)_V\ra(\PP\circ\tphi)_V$ for each $V\in\VN$ such that, whenever $\tilde{V}\subset V$, one has $\PP(i_{\tphi(\tilde{V})\tphi(V)})\circ\theta_V=\theta_{\tilde{V}}\circ\PP(i_{\tphi(\tilde{V})\tphi(V)})$.

In our case, an order automorphism $\tphi:\VN\ra\VN$ acts by pullback on the trivial presheaf $\POne$ over $\VN$ in the following way: for all $V\in\VN$, $(\POne\circ\tphi)_V=\POne_{\tphi(V)}=\{*\}$. 

Since each component $\POne_V$ of the trivial presheaf $\POne$ is just a one-element set, $\POne_V=\{*\}$, the only isomorphism $\theta_V:(\POne\circ\tphi)_V\ra(\POne\circ\tphi)_V$ is the identity map (for each $V\in\VN$). Hence, an automorphism of the trivial presheaf $\POne$ over $\VN$ is simply given by a shifting around of components, induced by an (order) automorphism of the base category $\VN$. We have shown:
\begin{lemma}	\label{Lem_AutomOfTrivPresh}
Let $\cN$ be a von Neumann algebra. There is a bijective correspondence between automorphisms $\Theta$ of the trivial presheaf $\POne$ over $\VN$ and order automorphisms $\tphi$ of the context category $\VN$.
\end{lemma}
From this and Thm.~\ref{Thm_DoeHar}, it follows:
\begin{corollary}
(Dye's theorem in presheaf form.) Let $\cN$ be a von Neumann algebra not isomorphic to $\bbC\oplus\bbC$ or to $M_2(\bbC)$. For every automorphism $\Theta$ of the trivial presheaf $\POne$ over the context category $\VN$, there is a unique (ultraweakly continuous) Jordan $*$-automorphism $\Phi:(\cN,\cdot)\ra(\cN,\cdot)$ such that $\tphi(V)=\Phi[V]$ for all $V\in\VN$, where $\tphi:\VN\ra\VN$ is the automorphism of the context category corresponding to $\Theta$.
\end{corollary}
Finally, from Lm.~\ref{Lem_AutomOfTrivPresh} and Thm.~\ref{Thm_WignerContextual} we obtain:
\begin{corollary}
    (Wigner's theorem in presheaf form.) Let $\cH$ be a Hilbert space, $\mathrm{dim}(\cH) \geq 3$. For every automorphism $\Theta$ of the trivial presheaf $\POne$ over the context category $\VH$, there is a unique unitary or anti-unitary operator $u$ such that
\begin{align*}
			\forall V\in\VN: \tphi(V) = uVu^*
\end{align*}
is the order automorphism of $\VH$ inducing $\Theta:\POne\ra\POne$. Conversely, every unitary or anti-unitary operator $u$ induces an automorphism $\Theta$ of the trivial presheaf $\POne$ over $\VN$.
\end{corollary} 
Note that the trivial presheaf contains exactly the right amount of information: every automorphism of $\POne$ gives a unitary or anti-unitary $u$ and vice versa. More physically speaking, every `rearrangement' of contexts that preserves the order (i.e., preserves how contexts are contained within each other) determines a unitary or anti-unitary operator and vice versa.

\section{The Kochen-Specker theorem and the spectral presheaf}	\label{Sec_KS}

The Kochen-Specker theorem \cite{KocSpe67} is deeply connected to contextuality. The usual interpretation of the theorem amounts to a negative statement of the kind `there are no non-contextual assignments of values to physical quantities'.\footnote{It is less clear which positive statement can be given and whether `contextual value assignments' exist or would even make sense.}

Kochen and Specker's result excludes certain state space models for quantum theory: a hypothetical state space $\Sigma$ on which each physical quantity $a$ is represented by a real-valued function $f_a:\Sigma\ra\bbR$ would imply the existence of valuation functions (i.e., assignments of values to physical quantities such that the spectrum rule and the functional composition principle hold, see below). In fact, each point $s$ of $\Sigma$ would provide a valuation function $v_s$, given by evaluation at that point, that is, the value of a physical quantity $a$ would simply be $v_s(a)=f_a(s)$. By showing that no valuation functions exist, Kochen and Specker exclude the existence of such a state space $\Sigma$.

Precisely because state space models are excluded, it is difficult to interpret the Kochen-Specker theorem in geometric terms. Also, it is not straightforward to see the exact nature of the connection between contextuality in our sense (commutative subalgebras of compatible physical quantities, arranged into a poset) and the nonexistence of valuation functions.

Both aspects were clarified by Isham and Butterfield in a beautiful series of papers \cite{IshBut98,IshBut99,HIB00,IshBut02}, with J. Hamilton as a co-author of the third paper. In fact, the context category first shows up in these papers, and so does the spectral presheaf $\Sig$, which will be defined below. The latter plays a central role in the topos approach to quantum theory \cite{DoeIsh11} and serves as a generalised state space for a quantum system, notwithstanding the Kochen-Specker theorem. In fact, the Kochen-Specker theorem is equivalent to the fact that the quantum state space $\Sig$ has no points (technically, it has no global sections). This reformulation serves as the prototype for the reformulations of the other fundamental theorems of quantum theory in this article.

In Sec.~\ref{Subsec_KSOld}, we will give a quick overview of the Kochen-Specker theorem and its background. Then, in Sec.~\ref{Subsec_KSIB}, we make some connections with contextuality and present the presheaf reformulation of the Kochen-Specker theorem by Isham, Butterfield, and Hamilton. Finally, we extend their results to von Neumann algebras.

\subsection{Valuation functions, the Kochen-Specker theorem, and contextuality}	\label{Subsec_KSOld}
In their seminal paper \cite{KocSpe67}, Kochen and Specker considered the question whether assignments of values to the physical quantities of a quantum system exist. Let $\cH$ be a separable Hilbert space, let $\BH$ be the algebra of bounded linear operators on $\cH$, and let $\BHsa$ be the real vector space of self-adjoint operators. A \emph{valuation function} is a function
\begin{align*}
			v:\BHsa &\lra \bbR\\
			a &\lmt v(a)
\end{align*}
such that
\begin{itemize}
	\item [(i)] for all $a\in\BHsa$, it holds that $v(a)\in\spec{a}$ (\emph{spectrum rule}),
	\item [(ii)] for all continuous functions $f:\bbR\ra\bbR$, it holds that $v(f(a))=f(v(a))$ (\emph{functional composition principle}).
\end{itemize}
Kochen and Specker show \cite{KocSpe67}:
\begin{theorem}	\label{Thm_KS}
    (Kochen-Specker theorem.) Let $\cH$ be a Hilbert space, $\dim\cH\geq 3$. There exist no valuation functions $v:\BHsa\ra\bbR$.
\end{theorem}
In the proof, a certain family of rays, i.e., rank-$1$ projections is considered. Each of these must be assigned either $0$ or $1$ according to the spectrum rule and in every orthogonal triple $p_1,p_2,p_3$ of rank-$1$ projections, exactly one projection is assigned $1$ and the others are assigned $0$. By carefully choosing the family of projections, Kochen and Specker construct an explicit counterexample: they show that no consistent assignment of values $0$, $1$ to the projections in their family is possible. The original proof used a configuration of $117$ rays in $\cH=\bbR^3$ (the real Euclidean space), which could later be reduced to $31$, and even fewer in $\bbC^4$. The proof of the result in real, three-dimensional Hilbert space implies the result in higher-dimensional, real and complex Hilbert spaces.

It is noteworthy that the proof of the Kochen-Specker theorem relies only on some basic projective geometry in $\bbR^3$. As soon as one accepts that there is a physical quantity, usually called `spin-$1$', which can be measured in all spatial directions and which takes three different values,\footnote{The rank-$1$ projections in the proof are projections onto eigenspaces of the spin-$1$ observable measured in various spatial directions.} the Kochen-Specker theorem holds unless one is willing to give up either the spectrum rule or the functional composition principle. Much less than the full mathematical apparatus of quantum mechanics is required to prove the theorem and very little physical interpretation, or metaphysical baggage, underlies the proof. 

Bell provided a proof of the same result, i.e., there are no non-contextual assignments of values to physical quantities, in \cite{Bel66}. His proof uses a continuity argument and Gleason's theorem and hence is not `discrete' as Kochen and Specker's proof.

\subsection{The spectral presheaf and reformulation of the Kochen-Specker theorem}	\label{Subsec_KSIB}
Kochen and Specker emphasise that an earlier proof of nonexistence of certain value assignments by von Neumann \cite{vNe32} was flawed, since von Neumann posed conditions on noncompatible physical quantities (represented by noncommuting self-adjoint operators), which Kochen and Specker regarded as unjustified. In contrast, the functional composition principle employed by Kochen and Specker seemingly is just a condition on commuting operators, since $a$ and $f(a)$ commute.

Yet, there are triples $\{a,b,c\}$ of self-adjoint operators such that
\begin{align*}
			c=f(a)=g(b)\; ,
\end{align*}
i.e., $c$ is both a function of $a$ and a function of $b$. Such a triple is called a \emph{Kochen-Specker triple.} Crucially, $c=f(a)$ commutes with $a$ and $c=g(b)$ commutes with $b$, but $a$ need not commute with $b$. In this way, the functional composition principle does pose conditions on noncommuting operators, too: since $f(v(a))=v(f(a))=v(g(b))=g(v(b))$, the (hypothetical) values $v(a)$ and $v(b)$ assigned to $a$ respectively $b$ are not independent.

The relations induced by Kochen-Specker triples led Isham and Butterfield to introduce the \emph{spectral presheaf}. In \cite{IshBut98} the self-adjoint operators themselves served as `stages', and in \cite{HIB00}, the step to commutative subalgebras of $\BH$ and the context category was taken. We will focus on the latter.

First, consider a single commutative subalgebra $V$ of $\BH$. We assume for the moment that $V$ is closed in the so-called norm topology, hence $V$ is a \emph{$C^*$-algebra}.\footnote{For details on the norm topology and $C^*$-algebras, see e.g. \cite{KadRin83}. In finite dimensions, any subalgebra of $\BH$ is a $C^*$-algebra.} Not surprisingly, there are valuation functions on $\Vsa$, the self-adjoint operators in $V$: every \emph{character}, that is, every multiplicative linear functional of norm $1$,
\begin{align*}
			\ld: \Vsa &\lra \bbR\\
			a &\lmt \ld(a)\; ,
\end{align*}
fulfils both the spectrum rule and the functional composition principle. Conversely, every valuation function is a character of $V$. Recall from Sec.~\ref{Subsec_Algebra} that multiplicative linear functionals of norm $1$ on a commutative von Neumann algebra $V$ are exactly the pure states of $V$. Hence, we consider the set
\begin{align*}
			\Sigma(V):=\{\ld:\Vsa\ra\bbR \mid \ld\text{ is a multiplicative linear functional of norm }1\}
\end{align*} 
of characters (or pure states) of $V$, traditionally called the \emph{Gelfand spectrum} of $V$.\footnote{The Gelfand spectrum $\Sigma(V)$ is equipped with the \emph{Gelfand topology}, which is the topology of pointwise convergence.} In physical terms, $\Sigma(V)$ is the (pure) state space of the physical system described by the physical quantities in $V$. As expected, the points of the state space $\Sigma(V)$ correspond exactly with valuation functions on $\Vsa$.

If $\tilde{V}\subset V$ is a unital $C^*$-subalgebra, then every character of $\tilde{V}$ arises as the restriction of some character of $V$, that is, there is a surjective map
\begin{align*}
			\Sigma(V) &\lra \Sigma(\tilde{V})\\
			\ld &\lmt \ld|_{\tilde{V}}\; .
\end{align*}
Hence, there is a canonical map from the state space of the bigger algebra $V$ to the state space of the smaller algebra $\tilde{V}$. Every valuation function on $V$ can be restricted to a valuation function on $\tilde{V}$.

In infinite dimensions, it is useful to work with commutative von Neumann subalgebras instead of the more general $C^*$-subalgebras, and we will do so from now on. This guarantees the existence of sufficiently many projection operators within our algebras and simplifies some arguments. In finite dimensions, there is no difference.

Isham, Butterfield, and Hamilton's key idea \cite{IshBut98,HIB00} was to combine all the state spaces for commutative subalgebras of a quantum system into one global object. This is the spectral presheaf:

\begin{definition}	\label{Def_SpecPresh}
Let $\cH$ be a Hilbert space. The \emph{spectral presheaf $\Sig$} of the algebra $\BH$ is the presheaf over the context category $\VH:=\mc{V}(\BH)$ given
\begin{itemize}
	\item [(i)] on objects: for all commutative von Neumann subalgebras $V\in\VH$, let
	\begin{equation*}
	    \Sig_V = \Sigma(V),\quad \text{the Gelfand spectrum of $V$,}
	\end{equation*}
	\item [(ii)] on arrows: for all inclusions $i_{\tilde{V}V}:\tilde{V}\hra V$, let
	\begin{align*}
				\Sig(i_{\tilde{V}V}):\Sig_V &\lra \Sig_{\tilde{V}}\\
				\ld &\lmt \ld|_{\tilde{V}}\; .
	\end{align*}
\end{itemize}
\end{definition}
It is clear by construction that the spectral presheaf $\Sig$ is a kind of state space for the quantum system, built from all the state spaces $\Sigma(V)$ of the commuting, compatible parts $V\in\VH$ of the noncommutative algebra $\BH$ of physical quantities.

As we saw in Sec.~\ref{Subsec_Contexts}, for a presheaf, the analogue of a point is a global section. What would a global section of the spectral presheaf $\Sig$ be? For every context $V\in\VH$, we have to pick one element $\ld_V\in\Sig_V$, the Gelfand spectrum of $V$. $\ld_V$ is a valuation function for the physical quantities in $V$, i.e., it assigns a value $\ld_V(a)$ to all $a\in\Vsa$ such that the spectrum rule and functional composition hold.

Moreover, if $\dot{V}$ is another commutative subalgebra that contains $a$, then $a$ is also contained in $\tilde{V}:=V\cap\dot{V}$. The value we assign to $a$ in $V$ is $\ld_V(a)$ and the value we assign to $a$ in $\dot{V}$ is $\ld_{\dot{V}}(a)$. Moreover, the value we assign to $a$ in $\tilde{V}=V\cap\dot{V}$ is
\begin{align*}
			\ld_{\tilde{V}}(a)=\ld_V|_{\tilde{V}}(a)=\ld_V(a)
\end{align*}
and also
\begin{align*}
			\ld_{\tilde{V}}(a)=\ld_{\dot{V}}|_{\tilde{V}}(a)=\ld_{\dot{V}}(a)\; ,
\end{align*}
hence
\begin{align*}
			\ld_{\tilde{V}}(a)=\ld_V(a)=\ld_{\dot{V}}(a)\; .
\end{align*}
The structure of a global section therefore guarantees that the value assigned to a physical quantity, represented by the self-adjoint operator $a$, is the same, independent of the context in which it lies. Since also the spectrum rule and the functional composition principle hold, every global section of $\Sig$ would provide a valuation function on all of $\BH$. Conversely, a valuation function would give a global section of $\Sig$.

Since the Kochen-Specker theorem shows that there are no valuation functions, i.e., no non-contextual value assignments, Isham, Butterfield, and Hamilton could give the following reformulation:
\begin{theorem}
The Kochen-Specker theorem is equivalent to the fact that the spectral presheaf $\Sig(\mc{H})$ has no global sections whenever $\mathrm{dim}(\cH) \geq 3$.
\end{theorem}
In more physical terms, the Kochen-Specker theorem is equivalent to the fact that the quantum state space $\Sig$ has no points. This does not mean, however, that $\Sig$ is `empty', it still has plenty of subobjects (which are the presheaf analogue of subsets). One can just not `focus down' to points, which would be (nonexistent) microstates.

We note that the nonexistence of points/global sections is not just a consequence of Kochen-Specker, but is exactly equivalent. This shows that the context category and the spectral presheaf contain just the right amount of information and that the Kochen-Specker theorem indeed is encoded by our notion of contextuality.

The Kochen-Specker theorem was generalised to von Neumann algebras in \cite{Doering2004}:
\begin{theorem}
(Generalised Kochen-Specker theorem.) Let $\cN$ be a von Neumann algebra with no direct summand of type $I_2$. There are no valuation functions $v:\cNsa\ra\bbR$.
\end{theorem}
The condition that $\cN$ has no summand of type $I_2$ generalises the condition that $\dim\cH\geq 3$ in the original proof of the theorem.

Since we can easily define the spectral presheaf of a von Neumann algebra $\cN$ (simply replace $\BH$ by $\cN$ and $\VH$ by $\VN$ in Def.~\ref{Def_SpecPresh}), we also have the following reformulation of the generalised Kochen-Specker theorem:
\begin{corollary}
Let $\cN$ be a von Neumann algebra with no direct summand of type $I_2$, let $\VN$ be the context category of $\cN$, and let $\Sig$ be its spectral presheaf. The generalised Kochen-Specker theorem is equivalent to the fact that $\Sig$ has no global sections.
\end{corollary}

\section{Gleason's theorem and the probabilistic presheaf}	\label{Sec_Gleason}
\subsection{Gleason's theorem}
Gleason's theorem, proven in 1957 \cite{Gle57}, showed that the Born rule follows from very modest assumptions. Let $\cH$ be a Hilbert space of dimension $3$ or greater. Assume that there is a function assigning probabilities to projection operators,
\begin{align*}
			\mu:\PH &\lra [0,1]\; ,
\end{align*}
such that
\begin{itemize}
	\item [(a)] $\mu(1)=1$,
	\item [(b)]\label{Cond_FinAdd} for all $p,q\in\PH$, if $pq=0$, then $\mu(p+q)=\mu(p)+\mu(q)$.
\end{itemize}
Condition (a) is the obvious \emph{normalisation} condition and (b) is \emph{finite additivity} on mutually orthogonal projections. Clearly, if one aims to have any probabilistic formalism relating to projections (representing propositions about a quantum system), having such a function $\mu$ that assigns probabilities to projections is the minimal and natural requirement. There is a built-in non-contextuality condition: every projection $p$ lies in many different contexts, but $\mu$ assigns just one probability to $p$, independently of contexts.

There is an obvious strengthening of finite additivity to infinite families of mutually orthogonal projections, called \emph{complete additivity}:
\begin{itemize}
	\item [(b')]\label{Cond_ComplAdd} for any family $(p_i)_{i\in I}$ of mutually orthogonal projections (i.e., $p_ip_j=\delta_{ij}p_i$ for all $i,j\in I$), it holds that $\mu(\bjoin_{i\in I}p_i)=\sum_{i\in I}\mu(p_i)$.
\end{itemize}
Note that if the underlying Hilbert space $\cH$ is separable, then the index set $I$ is at most countable.\footnote{There is an obvious notion of \emph{countable additivity}, complete additivity reduces to countable additivity for a separable Hilbert space and to finite additivity for a finite-dimensional Hilbert space.}

Gleason showed the following, partly answering an earlier problem posed by Mackey:

\begin{theorem}	\label{Thm_Gleason}
(Gleason's theorem.) Let $\dim\cH\geq 3$. Given a completely additive probability measure $\mu$ on projections, there always exists a unique positive trace-class operator with trace $1$, written $\rho_\mu$, such that
\begin{align*}
			\forall p\in\PH: \mu(p)=\tr(\rho_\mu p)\; .
\end{align*}
\end{theorem}
In finite dimensions $\rho_\mu$ is nothing but a density matrix. As mentioned in Sec.~\ref{Subsec_Algebra}, this means that every completely additive probability measure on projections determines a unique normal state of $\BH$. Conversely, every normal state, equivalently every positive trace-class operator of trace $1$ (or, in finite dimensions, every density matrix) $\rho$ determines a unique completely additive probability measure by
\begin{align*}
			\mu_\rho:\PH &\lra [0,1]\\
			p &\lmt \tr(\rho p)\; . 
\end{align*}
Obviously, $\mu_{\rho_{\mu}}=\mu$ and $\rho_{\mu_{\rho}}=\rho$. Hence, Gleason's theorem justifies the use of density matrices and the Born rule in quantum mechanics. The condition that the Hilbert space is at least three-dimensional is essential and we will assume $\dim\cH\geq 3$ from now on.

In order to understand the power of Gleason's theorem, note that the definition of a completely additive probability measure $\mu:\PH\ra [0,1]$ only poses conditions on mutually orthogonal, hence, commuting projections.\footnote{In other words, for any family $(p_i)_{i\in I}$ of mutually orthogonal projections condition (b') above is a condition within a context $V$ that contains all the $p_i$, $i\in I$. If there are several contexts that contain all the $p_i$, it does not matter which context we consider, since probabilities are assigned directly to projections, independently of the contexts they lie in.} 

For simplicity, let us assume the Hilbert space $\cH$ is finite-dimensional, $\cH=\bbC^n$. Let $V$ be a context of $\BH=M_n(\bbC)$, the complex $n\times n$-matrix algebra. Let $\{p_1,...,p_m\}$ denote the unique set of minimal projections in $V$. Then the $p_i$ are mutually orthogonal and $V$ is generated by them, $V=\{p_1,...,p_m\}''$. Every self-adjoint operator $a\in \Vsa$ is a unique real linear combination of the $p_i$, that is, $a=\sum_{i=1}^m A_i p_i$. We extend the finitely additive probability measure $\mu$ to a function $\mu:\Vsa\ra\bbR$ by
\begin{align*}
			\forall a\in \Vsa: \mu(a):=\sum_{i=1}^m A_i \mu(p_i)\; .
\end{align*}
This implies directly that if $r\in\bbR$, then $\mu(ra)=r\mu(a)$ and if $a,b\in \Vsa$, then $\mu(a+b)=\mu(a)+\mu(b)$. Hence, $\mu:\Vsa\ra\bbR$ is a real-linear function. Importantly, if $a$ and $b$ do not commute, it is not obvious at all if $\mu(a+b)=\mu(a)+\mu(b)$ holds or not. A finitely additive probability measure $\mu:\PH\ra [0,1]$ gives a function $\mu:\Vsa\ra\bbR$ that is \emph{linear in every context} $V$ in a straightforward way, but it is not clear initially why this function should also be linear across contexts, i.e., on noncommuting operators. Traditionally, a function that is linear on commuting operators is called \emph{quasi-linear}.

Gleason's result shows that there always exists a density matrix $\rho_\mu$ such that $\mu(p)=\tr(\rho_\mu p)$ and clearly, we also have
\begin{align*}
			\forall a\in\Vsa: \mu(a)=\tr(\rho_\mu a)
\end{align*}
due to linearity of the trace. Crucially, the map
\begin{align*}
			\tr(\rho_\mu \_):\BHsa &\lra \bbR\\
			a &\lmt \tr(\rho_\mu a)
\end{align*}
is linear on \emph{all} (self-adjoint) operators,
\begin{align*}
			\forall a,b\in\BHsa : \tr(\rho_\mu(a+b))=\tr(\rho_\mu a)+\tr(\rho_\mu b)\; ,
\end{align*}
which implies that $\mu$ is also linear on all operators. Hence, the quasi-linear function $\mu$ is in fact linear. In this way, Gleason's theorem solves a local-to-global problem, where `local' here means `on commuting operators' (or `within contexts') and global means `on all operators'.

By the efforts of many people, Gleason's theorem has been generalised to von Neumann algebras (see \cite{Mae89} and references therein):

\begin{theorem}	\label{Thm_GenGleason}
(Generalised Gleason's theorem.) Let $\cN$ be a von Neumann algebra with no direct summand of type $I_2$, and let $\mu:\PN\ra [0,1]$ be a finitely additive probability measure on the projections of $\cN$. There exists a unique state $\rho_\mu$ of $\cN$ such that
\begin{align*}
			\forall p\in\PN:\mu(p)=\rho_\mu(p)\; .
\end{align*}
\end{theorem}
Note that here $\rho_\mu:\cN\ra\bbC$ denotes the state itself (i.e., a positive linear functional of norm $1$), while before the state was denoted $\tr(\rho_\mu \_):\BH\ra\bbC$ and $\rho_\mu$ was just the positive trace-class operator (or density matrix). The reason is that the state $\rho_\mu$ need not be normal and hence there may be no density matrix. In fact, $\rho_\mu$ is normal if and only if the probability measure $\mu$ is completely additive.

\subsection{The probabilistic presheaf and reformulation of Gleason's theorem}
In order to relate Gleason's theorem (in its generalised form), Thm. \ref{Thm_GenGleason}, more explicitly to contextuality, we consider a certain presheaf that encodes probability assignments to projections. The obvious definition is:
\begin{definition}\label{def: probabilistic presheaf}
Let $\cN$ be a von Neumann algebra with context category $\VN$. The \emph{probabilistic presheaf $\PPi$} of $\cN$ over $\VN$ is the presheaf given
\begin{itemize}
	\item [(i)] on objects: for all $V\in\VN$, let
	\begin{align*}
				\PPi_V:=\{\mu_V:\PV\ra [0,1] \mid \mu_V\text{ is a finitely additive probability measure}\}\; ,
	\end{align*}
	\item [(ii)] on arrows: for all inclusions $i_{\tilde{V}V}:\tilde{V}\hra V$, let
	\begin{align*}
				\PPi(i_{\tilde{V}V}): \PPi_V &\lra \PPi_{\tilde{V}}\\
				\mu_V &\lmt \mu_V|_{\tilde{V}}\; .
	\end{align*}
	Here, the restriction $\mu_V|_{\tilde{V}}$ of the function $\mu_V:\PV\ra [0,1]$ to $\mc{P}(\tilde{V})\subset\PV$ is simply marginalisation.
\end{itemize}
\end{definition}

Note that this is the simplest possible definition of a presheaf built from finitely additive probability measures (FAPMs) on contexts.
An element $\mu_V\in\PPi_V$ is a FAPM for the projections in $V$, so it only assigns probabilities to projections in $V$, not to all projections (unlike the FAPM $\mu:\PN\ra [0,1]$ in the generalised Gleason's theorem, Thm. \ref{Thm_GenGleason}, which assigns probabilities to all projections in $\cN$).

The probabilistic presheaf $\PPi$ can be seen as a generalisation of the spectral presheaf $\Sig$ in the following way: at each context $V\in\VN$, the component $\Sig_V$ of $\Sig$ is the set of pure states $\ld:V\ra\bbC$, see Def.~\ref{Def_SpecPresh}. In the probabilistic presheaf $\PPi$, on the other hand, the component $\PPi_V$ is given by finitely additive probability measures $\mu_V:\PV\ra [0,1]$. The latter are positive linear functionals on $V$ of norm $1$, i.e., convex linear combinations of elements in $\Sig_V$.
In other words, the elements of $\PPi_V$ correspond to mixed states of $V$, while the elements of $\Sig_V$ correspond to pure states of $V$, equivalently, extreme points of $\PPi_V$.

What about global sections of the probabilistic presheaf? Prima facie, we do not know whether global sections exist or not, but we now show that every quantum state $\rho:\cN\ra\bbC$ gives a global section $\ga_\rho$ of $\PPi$. Define
\begin{align*}
			\forall V\in\VN:\ \ga_\rho(V):=(\rho|_{\PV}:\PV\ra [0,1])\; .
\end{align*}
Here, $\rho|_{\PV}$ is the restriction of the quantum state to the projections in the context $V$. Since $\rho$ is linear, $\rho|_{\PV}$ is a finitely additive probability measure on the projections in $V$. If a projection $p$ is contained in a context $V$ and a subcontext $\tilde{V}\subset V$, then
\begin{align*}
			\ga_\rho(\tilde{V})(p)=\rho|_{\mc{P}(\tilde{V})}(p)=\rho|_{\PV}(p)=\ga_\rho(V)(p)\; ,
\end{align*}
so $\ga_\rho$ is indeed a global section.

Conversely, let $\ga$ be a global section of the probabilistic presheaf $\PPi$. In every context $V\in\VN$ we have a FAPM $\ga(V):\PV\ra [0,1]$ on the projections of $V$. The restriction maps $\PPi(i_{\tilde{V}V})$ guarantee that, whenever a context $\tilde{V}$ is contained in another context $V$, a projection $p$ is assigned the same probability, no matter whether we regard $p$ as a projection in $V$ or in $\tilde{V}$. Hence, a global section $\ga$ of the probabilistic presheaf $\PPi$ gives a finitely additive probability measure $\mu$ on \emph{all} projections in $\cN$. By Gleason's theorem for von Neumann algebras, this determines a unique state $\rho_\gamma$ of the algebra $\cN$, provided $\cN$ has no type $I_2$ summand.\footnote{The latter condition is akin to the condition that the Hilbert space must be at least three-dimensional.}

Before we state Gleason's theorem in its contextual reformulation, we discuss the following slight variation of Def.~\ref{def: probabilistic presheaf}. Note that we may interpret the probability measures $\mu_V \in \PPi_V$ as positive operator-valued measures. In fact, by Gelfand duality every commutative von Neumann algebra $V \in \VN$ corresponds with an (extremely disconnected) compact Hausdorff space, whose $\sigma$-algebra of open (and closed) sets corresponds with the projection lattice $\mc{P}(V)$. Identifying $\mathbb{R}$ with (real-valued) $(1\times 1)$-matrices, $\mu_V$ (trivially) becomes a positive operator-valued measure, and by Naimark's theorem \cite{Naimark1943}, we can find a dilation of the form $\mu_V = v_V^* \varphi_V v_V$, where $v_V: \mathbb{C} \rightarrow \mc{K}$ is a bounded linear map into some Hilbert space $\mc{K}$, (equivalently, $v_V \in \mc{K}$ under scalar multiplication), and $\varphi_V: \mc{P}(V) \rightarrow \mc{P}(\mc{K})$ is an embedding (or spectral measure).

In this reading, we obtain a finitely additive probability measure $\mu_V$ by setting $v_V|_{\tilde{V}} = v_{\tilde{V}}$ and $\varphi_V|_{\tilde{V}} = \varphi_{\tilde{V}}$ whenever $\tilde{V} \subset V$: first, by Dye's theorem \cite{Dye55,BunceWright1993}, the latter defines an orthomorphism $\varphi: \PN \rightarrow \mc{P}(\mc{K})$, which lifts to a unique Jordan $*$-homomorphism $\Phi: \PN \rightarrow \mc{B}(\mc{K})$; second, by an extended version of Gleason's theorem in \cite{BunceWright1992}, $(v_V)_{V \in \VN}$ corresponds with a unique vector $v \in \mc{K}$. Note that in contrast to the marginalisation constraints in the probabilistic presheaf, it is crucial to restrict the dilations $\mu_V = v_V^* \varphi_V v_V$ with respect to both $v_V$ and $\varphi_V$, since there is (at least) a freedom in choosing complex phases $v_V \rightarrow e^{i\alpha}v_V$, $\alpha \in \mathbb{R}$ in every context, which leaves the measures $\mu_V$ invariant, but obscures linearity of $v \in \mc{K}$. Alternatively,
we may choose $v \in \mc{K}$ fixed and only consider global sections arising under restrictions with respect to $\varphi_V$ along context inclusion. This is the approach taken in Def.~\ref{def: dilated probabilistic presheaf} below.

Importantly, collections of dilations over contexts $(\mu_V = v^* \varphi_V v)_{V \in \VN}$ still correspond with quantum states $\rho = v^* \Phi v$. In finite dimensions, the latter is easily recognised as a purification of $\rho$. In particular, restricting to pure states,
we may choose $\mc{K} = \mc{H}$ and $\varphi: \PH \rightarrow \PH$ the identity map such that $|v\rangle \in \mc{H}$ is the pure state corresponding to $\rho(p) = \langle v| p |v\rangle = \tr(|v \rangle \langle v| p)$ for all $p \in \PN$. As always, mixed states correspond to convex combinations of pure states; in this sense, applying Naimark's theorem in contexts amounts to a type of intrinsic convexity condition with respect to the set of pure states. Taking the latter into account, we refine Def.~\ref{def: probabilistic presheaf} as follows.

\begin{definition}\label{def: dilated probabilistic presheaf}
    Let $\cN$
    be a von Neumann algebra with context category $\VN$ and $\mc{K}$ a Hilbert space.
    The \emph{dilated probabilistic presheaf $\PPi: \VN^\mathrm{op} \rightarrow \mathbf{Set}$} of $\cN$ over $\VN$ is the presheaf given
    \item [(i)] on objects: for all $V\in\VN$, let
	\begin{align*}
        \PPi_V &:= \{\mu_V: \mc{P}(V) \rightarrow [0,1] \mid \mu_V = v^* \varphi_V v \text{ for } v \in \mc{K}, \varphi_V: \mc{P}(V) \hookrightarrow \mc{P}(\mc{K}), \text{ and } \mu_V(1) = 1\}\; ,
	\end{align*}
	\item [(ii)] on arrows: for all inclusions $i_{\tilde{V}V}:\tilde{V}\hra V$, let
	\begin{align*}
		\PPi(i_{\tilde{V}V}): \PPi_V &\lra \PPi_{\tilde{V}},\\
		v^* \varphi_V v = \mu_V &\lmt \mu_{\tilde{V}} = v^* \varphi_V|_{\tilde{V}} v\; .
	\end{align*}
\end{definition}

Note that in a slight abuse of notation we denote the dilated probabilistic presheaf $\PPi$ of $\cN$ over $\VN$ by the same symbol as the probabilistic presheaf, $\PPi$ of $\cN$ over $\VN$. The reason is that, by the preceding discussion, both share the same set of global sections, in fact, we have shown the following.

\begin{theorem}	\label{Thm_GenGleasonContextual}
    (Generalised Gleason's theorem in contextual form.) Let $\cN$ be a von Neumann algebra with no direct summand of type $I_2$. There is a bijective correspondence between quantum states, that is, states on $\cN$, and global sections of the (dilated) probabilistic presheaf $\PPi$ over $\VN$.
\end{theorem}

This is our reformulation of Gleason's theorem, which connects it explicitly with contextuality. In contrast to the spectral presheaf, the (dilated) probabilistic presheaf does have global sections and they correspond exactly with quantum states. It is remarkable that the very simple definition of the probabilistic presheaf $\PPi$, with FAPMs in every context, connected by the obvious restriction maps in the form of marginalisation, suffices to guarantee this (for single systems).\footnote{In the next section we will see that FAPMs have to be refined to dilations in contexts, as defined in Def.~\ref{def: dilated probabilistic presheaf}, in order to guarantee a similar correspondence also in the multipartite case.}

As usual, the power of the construction lies in the restriction maps (and, of course, Gleason's theorem). In particular, no further local or global data is needed. In physical terms, there is no need for hidden variables. More importantly, there is no room for hidden variables: as soon as a theory assigns probabilities to all projections in dimension $3$ or greater in the obvious way, i.e., finitely additively on orthogonal projections, there exists a quantum state that provides this assignment of probabilities.

\begin{center}
			\textbf{There are no other (finitely additive) assignments of probabilities to projections apart from those given by quantum states.}
\end{center}

\noindent Any hidden variables or other extra data could at best give further restrictions.\footnote{Note that it is key that probabilities are assigned to \emph{all} projections in a finitely additive way. If one considers only a subset of projections, then other assignments of probabilities than those coming from quantum states are possible, e.g. PR-boxes \cite{PRbox}.} It is worthwhile mentioning the case $\cN=\BH$ explicitly. 
\begin{corollary}
    (Gleason's theorem in contextual form.) Let $\BH$ be the algebra of all bounded operators on a Hilbert space $\cH$. If $\dim\cH\geq 3$, then there is a bijective correspondence between quantum states, that is, states on $\BH$, and global sections of the (dilated) probabilistic presheaf $\PPi$ over $\VH$.
\end{corollary}

Moreover, one can consider a presheaf $\tPPi$ that is closely related to the (dilated) probabilistic presheaf $\PPi$, but has as component at $V\in\VN$ only completely additive probability measures. It is easy to check that if a von Neumann algebra $\cN$ has no type $I_2$ summand, there is a bijective correspondence between normal states on $\cN$ and global sections of the normal (dilated) probabilistic presheaf $\tPPi$ over $\VN$. As a special case, if $\dim\cH\geq 3$, there is a bijective correspondence between normal states on $\BH$ and global sections of $\tPPi$ over $\VH$.

The fact that Gleason's theorem is closely linked with contextuality in this manner was first observed in \cite{Doering2004}, made more explicit by de Groote \cite{deG07}, and in a form very similar to the one above in \cite{Doe12b}.


\section{Bell's theorem and contextuality}\label{Sec_Bell}

Bell's seminal paper \cite{Bell1964} responds to a long-standing conjecture by Einstein, Podolsky and Rosen (EPR) \cite{EPR}, who claim quantum theory is only a statistical version of a more fundamental theory, similar to the relation between thermodynamics and statistical mechanics. Besides the probabilistic nature of quantum theory, this idea is motivated by certain nonlocal features present in the quantum formalism, believed to be resolved within the more fundamental theory. As a response to EPR's famous thought experiment, Bell formalises EPR's assumption of an underlying space of hidden variables and derives a constraint for the maximal amount of correlations possible in such theories under the additional assumption of locality \cite{Bell1964, BellSpeakable}. However, some quantum mechanically predicted and experimentally verified correlations \cite{AspectEtAl1981, ZeilingerEtAl2015, ShalmEtAl2015} do not obey these constraints and thus cannot be reproduced by any local hidden variable model. We show that, as with the other theorems discussed in this article, the essence of Bell's theorem is naturally encoded in a partial order of contexts and we discuss the relation between contextuality and locality in this setting. The connection between these concepts has been highlighted before \cite{AbramskyBrandenburger2011}, here we extend these results in several ways, in particular, we stress the importance of composition.

We first recall the derivation of Bell's theorem in Sec.~\ref{sec: Correlation in Classical Theories} emphasising the assumption of an underlying single-context state space \cite{Edwards1979}, i.e., with trivial physical contextuality. In Sec.~\ref{sec: Contextuality, Composition and Locality} we show that factorisability is closely linked with composition of (single-context) state spaces. We then consider two further notions of composition, one via contexts and the standard composition in quantum theory, via tensor products. Although standard composition results in a much richer context structure than our notion of composition via contexts, we show that it suffices to consider the dilated probabilistic presheaf over the smaller set of product contexts (to be defined as the \emph{Bell presheaf} below), in order to uniquely single out quantum correlations. This is interpreted as a contextual form of Bell's theorem.



\subsection{Correlations in classical theories}\label{sec: Correlation in Classical Theories}


\subsubsection{Classical state spaces}\label{sec: Classical State Spaces}

We first give an account of what we mean by a \emph{classical} theory. For our purposes it will be enough to consider the kinematics and so we start with a set (soon to be upgraded to an algebra) of observables $\mathcal{O}$. We take as a defining property of a classical theory that all its observables are \emph{simultaneously measurable}, from the perspective of physical contextuality we are thus considering the trivial case of a single context \cite{Edwards1979}.\footnote{We are not considering additional non-commuting operators such as in the Koopman-von Neumann formalism for classical mechanics (see \cite{Morgan2020}, for instance).} Observables $a \in \mathcal{O}$ in classical theories are mathematically represented by measurable\footnote{Here, `measurable' is used in the mathematical sense.} functions $f_a: \Sigma \rightarrow \bbR$ from some measure space $(\Sigma,\sigma,ds)$ to the real numbers. $\Sigma$ is called the \emph{(single-context) state space} of the theory and every microstate $s \in \Sigma$ assigns truth values to propositions of the form $'a \in \Delta'$ (read \emph{`the physical quantity $a$ has a value within the Borel subset $\Delta \in \mathbb{R}$'}):
\begin{equation}\label{eq: truth values}
    \Theta('a \in \Delta',s) := \begin{cases}
        1 \ \mathrm{if} \ s \in f_a^{-1}(\Delta) \\
        0 \ \mathrm{otherwise}
    \end{cases}
\end{equation}
We can therefore speak of the \emph{value of an observable $v_s(a)$ given the state $s \in \Sigma$} in the intuitive sense, i.e., through evaluation of the corresponding measurable function,
\begin{equation}\label{eq: non-contextual value assignment}
    v_s(a) := f_a(s)\; .
\end{equation}
The \emph{valution functions} $v_s: \mathcal{O} \rightarrow \mathbb{R}$ in Eq.~(\ref{eq: non-contextual value assignment}), which we already discussed in connection with the Kochen-Specker theorem in Sec.~\ref{Sec_KS}, are defined on all observables, in other words, every observable has an intrinsic (sharp) value in every state.\footnote{Note that the spectrum rule, $v_s(a) \in \mathrm{sp}(a) = \mathrm{Im}(f_a)$, is trivially satisfied.} The observation that all observables \emph{simultaneously} take deterministic values justifies to model \emph{physical states} by points in some space $\Sigma$ and observables by (measurable) functions $f_a: \mc{O} \rightarrow \mathbb{R}$ in the first place. Of course, this inductive reasoning has to be revisited for non-classical theories, that is, theories with non-trivial physical contextuality. In so doing, we attribute a fundamental role to observables, whereas states appear as a secondary concept. This perspective will become important in Sec.~\ref{sec: General Locality Constraints and Composition}, when we go from single to multiple-context state spaces.

In (classical) physics it is natural to equip the set of observables $\mathcal{O}$ with the structure of an algebra. In fact, by modelling observables as functions we are automatically given a vector space structure as well as a product by pointwise multiplication of functions.\footnote{We include the `trivial' observable $e \in \mc{O}$ represented by the constant function $f_e = 1$. This observable simply asks the question `Is the system there?' and the answer is always `yes'.} It is straightforward to extend the definition of valuation functions in Eq.~(\ref{eq: non-contextual value assignment}) to this algebraic structure, namely, for all $a,b \in \mc{O}$, $r \in \mathbb{R}$ and $s \in \Sigma$ we set
\begin{equation}\label{eq: algebra homomorphism}
    v_s(a \cdot b) := f_a(s) \cdot f_b(s), \quad
    v_s(a + b) := f_a(s) + f_b(s), \quad
    v_s(r a) := r f_a(s)\; .
\end{equation}
In other words, classical states $s \in \Sigma$ correspond to \emph{algebra homomorphisms} $v_s: \mc{O} \rightarrow \mathbb{R}$.

Note that in the presence of physical contextuality this suggests to consider generalised classical states to be valuation functions, that is, \emph{partial} algebra homomorphisms for which Eq.~(\ref{eq: algebra homomorphism}) only holds within (sub)algebras of simultaneously measurable observables, or contexts. In the setting of von Neumann algebras, which we will adopt again in later sections, Eq.~(\ref{eq: algebra homomorphism}) holds as a consequence of the functional composition principle,
\begin{equation}\label{eq: KS FUNC principle}
    v_s(f(a)) = f(v_s(a))\; ,\footnote{Recall from Sec.~\ref{Subsec_KSOld}, that $f: \mathbb{R} \rightarrow \mathbb{R}$ is a continuous function and $a, f(a) \in \mc{N}_\mathrm{sa}$ are self-adjoint operators.}
\end{equation}
whenever $\mc{O}$ is a commutative von Neumann algebra. Yet, as the Kochen-Specker theorem shows, such valuation functions cannot exist under mild and natural conditions. We will see that Bell's theorem assumes a similar reformulation as a no-go-result for such (generalised) classical states, based on the additional assumption of composition.

In the remainder of this section we give a derivation of factorisability and thus Bell's theorem for classical, single-context theories with composition described by the canonical product of state spaces. Given two subsystems with measure spaces $(\Sigma_1,\sigma_1,ds_1)$ and $(\Sigma_2,\sigma_2,ds_2)$, the composite state space is defined as the cartesian product $\Sigma_{1\& 2} := \Sigma_1 \times \Sigma_2$ with product $\sigma$-algebra $\sigma_{1\& 2}$ generated by elements $B_1 \times B_2$, $B_1 \in \sigma_1$, $B_2 \in \sigma_2$, and product measure $ds_{1\& 2} := ds_1 \times ds_2$ satisfying the condition
\begin{equation}\label{eq: product state space}
    (ds_1 \times ds_2)(B_1 \times B_2) = ds_1(B_1) \cdot ds_2(B_2)\; .\footnote{A product measure always exists, it is also unique if the individual measures are $\sigma$-finite. Note also that the cartesian product extends to spaces with more structure, e.g. symplectic manifolds.}
\end{equation}
In a similar way we obtain composite state spaces with multiple subsystems.
Correspondingly, composite observables $a \in \mc{O}$ are represented by measurable functions $f_a: \Sigma \rightarrow \mathbb{R}^n$ on the composite state space $\Sigma = \times_{i=1}^n \Sigma_i$. The algebra $\mc{O}$ of observables of the composite system is generated by the algebras of its subsystems by taking real linear combinations and products (and suitable limits if we consider topologically closed algebras). Clearly, evaluation on elements $s \in \Sigma$ still yields algebra homomorphisms similarly to Eq.~(\ref{eq: algebra homomorphism}), hence, we obtain composite valuation functions $v_s: \mc{O} \rightarrow \mathbb{R}^n$ from the obvious generalisation of Eq.~(\ref{eq: non-contextual value assignment}) to composite observables.

In order to obtain a generalisation of the truth values in Eq.~(\ref{eq: truth values}) it is thus enough to consider tuples $\mathbf{a} = (a_1,\cdots,a_n) \in \mathcal{O}$ with $a_i \in \mathcal{O}_i$ for $i \in \{1,\cdots,n\}$ as well as measurable functions $f_\mathbf{a}: \Sigma \rightarrow \mathbb{R}^n$, $f_\mathbf{a}(s) := (f_{a_1}(s_1),...,f_{a_n}(s_n))$ with $s \in \Sigma = \times_{i=1}^n \Sigma_i$. Namely, we define the truth value of the proposition $'\mathbf{a} \in \Delta'$ with Borel set $\Delta := \times_{i=1}^n \Delta_i$ as follows:
\begin{equation}\label{eq: composite truth values}
    \Theta('\mathbf{a} \in \Delta',s) := \begin{cases}
        1 \ \mathrm{if} \ s \in f_\mathbf{a}^{-1}(\Delta) \\
        0 \ \mathrm{otherwise}
    \end{cases}
    = \begin{cases}
        1 \ \mathrm{if} \ s_i \in f_{a_i}^{-1}(\Delta_i) \ \forall i \\
        0 \ \mathrm{otherwise}
    \end{cases}
    = \prod_{i=1}^n \Theta('a_i \in \Delta_i',s_i)
\end{equation}


\subsubsection{Statistical mixtures and joint probability distributions}\label{sec: Statistical mixtures and joint probability distributions}

Spectrum rule, Eq.~(\ref{eq: non-contextual value assignment}), and functional composition, Eq.~(\ref{eq: KS FUNC principle}), are specific to pure states, mixed states on the other hand are modelled as statistical averages by means of probability distributions over the state space $p: \Sigma \rightarrow \mathbb{R}$,
\begin{equation*}
\int_\Sigma ds \ p(s) = 1, \quad p(s) \geq 0 \ \forall s \in \Sigma\; .
\end{equation*}
The probability for the event corresponding to the Borel set $\Delta \subset \mathbb{R}$ when measuring the observable $a \in \mathcal{O}$ of a system in the mixed state $p$ is given by
\begin{equation}\label{eq: statistical mixtures}
    P(\Delta \mid a) = \int_{\{s \in \Sigma \mid v_s(a) \in \Delta\}} ds \ p(s) = \int_{f_a^{-1}(\Delta)} ds \ p(s) = \int_\Sigma ds \ p(s) \ \Theta('a \in \Delta',s)\; .
\end{equation}
Note that in the last step we have used the indicator function $\Theta('a \in \Delta',s)$ in Eq.~(\ref{eq: truth values}). For instance, the probability for obtaining a particular outcome $A$ corresponds to the Borel set $\Delta_A := \{A\}$.
Analogously, for joint probability distributions on a bipartite system we have with Eq.~(\ref{eq: composite truth values}):
\begin{align}
    P(A,B \mid a,b) &= \int_\Sigma ds \ p(s) \ \Theta('\mathbf{a} \in (\Delta_A,\Delta_B)', s) \nonumber \\
    &= \int_\Sigma ds \ p(s) \ \Theta('{a \in \Delta_A}',s_1) \cdot \Theta('{b \in \Delta_B}',s_2) \label{eq: general factorisability}
\end{align}
Crucially, for classical systems the support of the joint probability distribution splits according to Eq.~(\ref{eq: composite truth values}).
Since furthermore $p(s_1,s_2 \mid \lambda) = p(s_1 \mid \lambda)p(s_2 \mid \lambda)$ over `patches' $S_1 \times S_2 \subseteq \Sigma$ such that $s_1 \in S_1 \subseteq \Sigma_1$, $s_2 \in S_2 \subseteq \Sigma_2$,\footnote{Note that since $P(s_1,s_2) = P(s_1)P(s_2 \mid s_1)$, this is trivially the case for $(s_1,s_2) = \lambda \in \Lambda = \Sigma_1 \times \Sigma_2$.} we may define the effective parameter space $\Lambda$\footnote{The existence of such a set $\Lambda$
is famously postulated by Reichenbach's principle of the common cause.
Conversely, having defined composition of systems in terms of their (classical) state spaces, probability distributions naturally admit such a set $\Lambda$ and moreover decompose according to Eq.~(\ref{eq: Bell locality}).}
as a set of such patches covering $\Sigma$, yielding (continuing from above)
\begin{align}
    &= \int_\Lambda d\lambda \ p(\lambda) \int_{(S_1 \times S_2)(\lambda)} ds \ p(s_1,s_2 \mid \lambda) \ \Theta('{a \in \Delta_A}',s_1) \cdot \Theta('{b \in \Delta_B}',s_2) \nonumber \\
    &= \int_\Lambda d\lambda \ p(\lambda) \int_{S_1(\lambda)} ds_1 \ p(s_1 \mid \lambda) \ \Theta('{a \in \Delta_A}',s_1) \cdot \int_{S_2(\lambda)} ds_2 \ p(s_2 \mid \lambda) \ \Theta('{b \in \Delta_B}',s_2) \nonumber \\
    &= \int_\Lambda d\lambda \ p(\lambda) \ P(A \mid a,\lambda) \cdot P(B \mid b,\lambda)\; , \label{eq: Bell locality}
\end{align}
which is the standard form of factorisability.

In this reading, the splitting of classical joint probability distributions according to factorisability fundamentally stems from the splitting of supports of indicator functions $\Theta('a \in \Delta',s)$, which follows by the existence of local (single-context) state spaces with composition defined by the cartesian product.
We have thus derived Eq.~(\ref{eq: Bell locality}) from essentially two assumptions:
\begin{enumerate}
	\item [(a)] trivial physical contextuality,  i.e., just a single context (in each subsystem) and
	\item [(b)] the cartesian product of state spaces as the state space of the composite system.
\end{enumerate}
Since condition (b) is entirely natural for single-context state spaces, Eq.~(\ref{eq: Bell locality}) can also be read as a consequence of \emph{just} trivial physical contextuality.\\

The locality constraint in factorisable distributions is simply the condition that the joint probability distribution is a statistical average over products of local distributions, which depend on local data (observables and outcomes) only. By modelling the composite system via the cartesian product of state spaces this is automatic---neither choice nor outcome of an observable affect the other factor in the product. Factorisability thus corresponds to composition given by the cartesian product and by the above argument to (trivial) physical contextuality. This argument then suggests an intimate relationship between the following concepts:
\begin{equation*}
    \mathbf{contextuality\  \longrightarrow \ composition\  \longrightarrow \ locality\ }
\end{equation*}
Of course, at this point a relation only exists in the very special case of trivial physical contextuality in classical systems. Nevertheless, in Sec.~\ref{sec: General Locality Constraints and Composition} we will see how these concepts are closely related also in the multiple-context setting.\\

\subsection{Contextuality, composition and locality}\label{sec: Contextuality, Composition and Locality}

Note that the derivation in Sec.~\ref{sec: Statistical mixtures and joint probability distributions} crucially depends on the assumption of underlying classical state spaces with composition defined in terms of the cartesian product. In such systems all observables are simultaneously measurable, which means they are trivial from the perspective of physical contextuality. Clearly, this is not the situation we are facing in quantum theory, where the Kochen-Specker theorem, Thm.~\ref{Thm_KS}, rules out a classical state space picture.

Therefore, shifting perspective from states to observables, in Sec.~\ref{sec: General Locality Constraints and Composition} we will discuss alternative ways to define composition, in particular, we motivate composition of systems based on observables and the order of contexts. We study the implications of this kind of composition based on context structure for the \emph{Bell presheaf}, i.e., the dilated probabilistic presheaf over product contexts in Sec.~\ref{sec:Bell Presheaf and Composition of Contexts}.


\subsubsection{Locality constraints and composition}\label{sec: General Locality Constraints and Composition}

\textbf{Composition via cartesian products of state spaces.} We consider the notion of composition in more detail. Recall that we defined composition of classical systems in terms of their state spaces, namely via the product of the corresponding measure spaces. On the other hand, observables in classical theories are represented by measurable functions and every measurable function on the composite state space can be approximated by suitable limits of linear combinations of indicator functions (cf. Eq.~(\ref{eq: composite truth values})). In this sense, it does not matter whether we define composition in terms of states or observables for classical systems. In fact, if we take classical systems to be given by commutative von Neumann algebras $\cN_i$, $i=1,2$ with corresponding state spaces given by Gelfand spectra $\Sigma_i = \Sigma(\cN_i) \simeq \Gamma(\Sig(\mc{V}(\cN_i)))$,\footnote{Note that the category of \emph{commutative} von Neumann algebras is equivalent to the category of localizable measurable spaces, that is, measurable spaces for which the Boolean algebra of equivalence classes modulo sets of measure zero is complete \cite{Segal_EquivalencesOfMeasureSpaces}.} this equivalence reads,
\begin{equation}\label{eq: equivalence between composition of states and contexts}
    \Sigma_{1\& 2} = \Sigma_1 \times \Sigma_2 \simeq \Gamma(\Sig(\mc{V}(\cN_1))) \times \Gamma(\Sig(\mc{V}(\cN_2))) = \Gamma(\Sig(\mc{V}(\cN_1) \times \mc{V}(\cN_2)))\; .
\end{equation}
Here, the final equality refers to the context product in Eq.~(\ref{eq: product context category}}) below.

By the Kochen-Specker theorem in contextual form, $\Gamma(\Sig(\VN))$ is empty whenever $\cN$ is a (noncommutative) von Neumann algebra (not of type $I_2$). The equivalence in Eq.~(\ref{eq: equivalence between composition of states and contexts}) thus breaks down for such algebras. Nevertheless, composition in terms of state spaces can be carried over to quantum systems if we define the state space of the composite system in terms of convex combinations of elements in the cartesian product of global sections of the probabilistic presheaves of subsystems instead:\footnote{Recall from Thm.~\ref{Thm_GenGleasonContextual} that global sections of the probabilistic presheaf bijectively correspond with quantum states.}
\begin{equation}\label{eq: product on global sections}
    \Gamma_{1\& 2} := \mathrm{Conv}(\Sigma_{1\& 2}), \quad \Sigma_{1\& 2} := \Gamma(\PPi(\mc{N}_1)) \times \Gamma(\PPi(\mc{N}_2))
\end{equation}
Note that $\Gamma_{1\& 2}$ is the set of separable states of the corresponding tensor product quantum system. It is thus easy to show that factorisability holds for \emph{any} systems---with or without local physical contextuality---as long as composition is defined in this way.

\begin{proposition}\label{prop: local physical contextuality implies Bell inequalities}
    Let $\cN_1$, $\cN_2$ be possibly noncommutative von Neumann algebras and let the set of states on the composite system, $\Gamma_{1\& 2}$, be defined according to Eq.~(\ref{eq: product on global sections}). Then all states in $\Gamma_{1\& 2}$ are factorisable and satisfy the Bell inequalities.
\end{proposition}

\begin{proof}
    Let $\rho \in \Gamma_{1\& 2}$. By definition, $\rho$ is a convex combination of product states $\gamma = (\gamma_1, \gamma_2) \in \Sigma_{1\& 2}$ with $\gamma_1 \in \Gamma(\PPi(\mc{N}_1))$ and $\gamma_2 \in \Gamma(\PPi(\mc{N}_2))$. Hence, there exists a measure $\mu_\rho: \Sigma_{1\& 2} \rightarrow \mathbb{R}^+_0$ such that for all $p \in \mc{P}(\mc{N}_1)$, $q \in \mc{P}(\mc{N}_2)$,
    \begin{equation}\label{eq: cartesian product on states implies factorisability}
        \rho(p,q)
        = \int_{\Sigma_{1\& 2}} d\mu_\rho(\gamma) \ \gamma_1(p) \cdot \gamma_2(q)
        = \int_{\Sigma_{1\& 2}} d\mu_\rho(\gamma) \ \tr(\rho_{\gamma_1} p) \cdot \tr(\rho_{\gamma_2} q)\; .
    \end{equation}
    The last equality follows by Gleason's theorem in contextual form, Thm.~\ref{Thm_GenGleasonContextual}. Eq.~(\ref{eq: cartesian product on states implies factorisability}) is just factorisability and the Bell inequalities thus necessarily hold.
\end{proof}

Clearly, this argument is not restricted to states on von Neumann algebras (density matrices in finite dimensions), but holds for arbitrary locally stochastic models with composition defined by the cartesian product similar to Eq.~(\ref{eq: product on global sections}). Every stochastic, factorisable model thus satisfies the Bell inequalities (cf. \cite{ClauserHorne1974}). Conversely, it is interesting to note that by \cite{Fine1982} the latter is equivalent to the existence of a deterministic local hidden variable model for the composite system. In this sense stochastic, factorisable models such as those with local physical contextuality, yet composition defined via Eq.~(\ref{eq: product on global sections}), still correspond to single-context state spaces.

Succinctly, factorisability is a direct consequence of composition defined in terms of the cartesian product of state spaces. 

\textbf{Composition via contexts.} It is clear from Prop.~\ref{prop: local physical contextuality implies Bell inequalities} and the above argument on stochastic, factorisable models (cf. \cite{Fine1982}) that we cannot use Eq.~(\ref{eq: product on global sections}) to define a suitable notion of composition for quantum systems. Shifting our focus from states to observables and their physical contextuality as encoded in the partial order of contexts, we can distill a second, different notion of composition from Eq.~(\ref{eq: equivalence between composition of states and contexts}) as follows:
\begin{equation}\label{eq: product context category}
    \mc{V}_{1\&2} := (\mc{V}_1 \times \mc{V}_2, \subset_{1\& 2}), \quad (\tilde{V}_1,\tilde{V}_2) \subset_{1\& 2} (V_1,V_2) \ :\Lra \ (\tilde{V}_1 \subset_1 V_1, \tilde{V}_2 \subset_2 V_2)\; .
\end{equation}
This is simply the product in the category of partial orders with composite contexts given by the cartesian product of local contexts and their natural product order $\subset_{1\& 2}$. 

With $\mc{V}_{1\&2}$ as base category we can build the probabilistic presheaf $\PPi(\mc{V}_{1\&2})$. Each component $\PPi_{(V_1,V_2)}$ consists of the finitely additive probability measures $\mu:\mc{P}(V_1 \otimes V_2)\ra[0,1]$ together with the obvious restriction maps. More importantly, we define the \emph{Bell presheaf} as the dilated probabilistic presheaf $\PPi(\mc{V}_{1\&2})$ according to Def.~\ref{def: dilated probabilistic presheaf}.

\begin{definition}\label{def: Bell presheaf}
    Let $\cN_1,\cN_2$ be a von Neumann algebras with context category $\mc{V}(\cN_1)$, $\mc{V}(\cN_2)$, respectively. Then we call the dilated probabilistic presheaf $\PPi(\mc{V}_{1\& 2})$ over the product context category $\mc{V}_{1\& 2} := \mc{V}(\cN_1) \times \mc{V}(\cN_2)$ the
    \emph{Bell presheaf of $\cN_1$ and $\cN_2$}.
\end{definition}

Clearly, there are versions of the Bell presheaf also for multipartite systems. We remark that composition via contexts is the most basic construction when considering contexts of the component systems (and a straighforward generalisation of Eq.~(\ref{eq: equivalence between composition of states and contexts})). We will have to show that this construction
has a meaningful physical interpretation.

\textbf{Composition via tensor products.} Before we explore the consequences of composition defined via contexts for the Bell presheaf, we end this section by mentioning a possible third way of defining composition, which in fact is the standard composition in quantum theory. 
There, the pure state space $\mc{S}(\cH)$ is the projective space corresponding to the Hilbert space $\cH$. Given component systems with Hilbert spaces $\cH_1,\cH_2$, the Hilbert space of the composite system is $\cH_1\otimes\cH_2$, hence,
\begin{equation*}
    \mc{S}_{1\& 2} := \mc{S}(\cH_1 \otimes \cH_2)\; .
\end{equation*}
Note that there are many more contexts for this kind of composition than for composition via contexts described above: the poset $\mc{V}(\cH_1\otimes\cH_2)$ contains many contexts that are not of the form $V_1\otimes V_2$, which are the only contexts available in the poset $\mc{V}_{1\&2}$. There is a functor
\begin{equation*}
			\mc{V}_{1\&2} \lra \mc{V}(\cH_1\otimes\cH_2), \quad \quad (V_1,V_2) \lmt V_1\otimes V_2
\end{equation*}
that is fully faithful, but not surjective on objects. We say that $\mc{V}_{1\&2}$ contains only \emph{product} (or \emph{twisted product}) contexts (cf. \cite{FreDoe19b}).

\subsubsection{The Bell presheaf and reformulation of Bell's theorem}\label{sec:Bell Presheaf and Composition of Contexts}

As we saw in the previous section, composition via contexts gives a much smaller poset of contexts, $\mc{V}_{1\&2}$, than the usual composition via tensor products, which leads to $\mc{V}(\cH_1\otimes\cH_2)$. From Gleason's theorem in contextual form, Thm. \ref{Thm_GenGleasonContextual}, we know that quantum states of the composite system described by $\cH_1\otimes\cH_2$ correspond bijectively with global sections of the probabilistic presheaf $\PPi(\mc{V}(\cH_1\otimes\cH_2))$ if $\dim\cH_1,\dim\cH_2\geq 3$.

Since $\mc{V}(\cH_1\otimes\cH_2)$ is a much richer poset than $\mc{V}_{1\&2}$, there are many more restriction maps in the (dilated) probabilistic presheaf $\PPi(\mc{V}(\cH_1\otimes\cH_2))$ than in the Bell presheaf $\PPi(\mc{V}_{1\&2})$. Each global section of $\PPi(\mc{V}(\cH_1\otimes\cH_2))$, that is, each quantum state, induces a global section of $\PPi(\mc{V}_{1\&2})$, but it is not clear \emph{a priori} whether the converse holds. The Bell presheaf $\PPi(\mc{V}_{1\&2})$ could potentially have many more global sections than those corresponding with quantum states. Remarkably, this is not the case. In order to see this, the following lemma is crucial, for details we refer to \cite{FreDoe19b}. For simplicity, we restrict the presentation to the case $\cN = \BH$ for finite-dimensional Hilbert spaces $\cH$;
we leave the general case (of arbitrary von Neumann algebras) for future work.

\begin{lemma}\label{lm: global sections to Jordan homos}
    Let $\cH_i$, $i=1,2$ be Hilbert spaces $\mathrm{dim}(\cH_i) \geq 3$ finite, $\mc{B}(\mc{H}_i)$ the algebra of physical quantities, and $\widetilde{\mc{V}(\cH_i)}$ the corresponding context categories. Then for every global section $\gamma \in \Gamma(\PPi(\mc{V}_{1\& 2}))$ of the Bell presheaf in Def.~\ref{def: Bell presheaf} there exists a unique linear map $\phi^\gamma: \mc{B}(\cH_1) \rightarrow \mc{B}(\cH_2)$. Moreover, there exists a Hilbert space $\mc{K}$, a linear map $v: \cH_1 \rightarrow \mc{K}$, and a Jordan $*$-homomorphism $\Phi^\gamma: \cJ(\mc{B}(\cH_1)) \rightarrow \cJ(\mc{B}(\mc{K}))$ such that
    \begin{equation*}
        \phi^\gamma = v^* \Phi^\gamma v\; .
    \end{equation*}
\end{lemma}

A related result by Wallach \cite{Wallach2000} shows that for finite-dimensional systems, frame functions over `unentangled'\footnote{`Unentangled' here refers to observables, which group into product contexts, and not to states.} bases uniquely correspond with self-adjoint operators and thus almost correspond with quantum states (in the form of density operators). In particular, there also exists a unique positive linear map $\phi^\gamma: \mc{B}(\mc{H}_1) \rightarrow \mc{B}(\mc{H}_2)$, yet generally not of the form in Lm.~\ref{lm: global sections to Jordan homos}. The difference is that \cite{Wallach2000} considers global sections of the \emph{undilated} probabilistic presheaf $\PPi(\mc{V}_{1\& 2})$ in Def.~\ref{def: probabilistic presheaf} (for details, see \cite{FreDoe19b}). In contrast, it is crucial to restrict to the Bell presheaf in Lm.~\ref{lm: global sections to Jordan homos}. In fact, by means of the latter the connection with states can be made precise as follows. In finite dimensions, every state corresponds with a density matrix. By Choi's theorem \cite{Choi1975}, every density matrix on the composite system $\cH_1 \otimes \cH_2$ corresponds with a completely positive, trace-preserving map $\phi: \mc{B}(\cH_1) \rightarrow \mc{B}(\cH_2)$. By Stinespring's theorem \cite{Stinespring1955}, every such completely positive map $\phi$ is of the form $\phi = v^*\Phi v$ with $v: \cH_2 \rightarrow \mc{K}$ a linear map and $\Phi: \mc{B}(\cH_1) \rightarrow \mc{B}(\mc{K})$ a *-homomorphism. In other words, a global section corresponds with a quantum state if and only if the Jordan *-homomorphism in Lm.~\ref{lm: global sections to Jordan homos} lifts to a *-homomorphism.

Not every Jordan *-homomorphism is also *-homomorphism, but it almost is. It turns out that for the special case of $\cN = \BH$ there are exactly two ways to lift a Jordan algebra to a von Neumann algebra: by augmenting the symmetric product (anticommutator) to an associative product $a \circ b = \frac{1}{2}\{a,b\} \pm \frac{1}{2}[a,b]$. Moreover, for every Jordan *-homomorphism $\Phi: \mc{J}(\mc{B}(\cH_1)) \rightarrow \mc{J}(\mc{B}(\mc{K}))$ it holds $\Phi([a,b]) = \pm [\Phi(a),\Phi(b)]$, however, $\Phi$ is only a *-homomorphism if it also preserves the commutator, $\Phi([a,b]) = [\Phi(a),\Phi(b)]$ for all $a,b \in \mc{B}(\cH_1)$ (cf. \cite{Kadison1951}). The sign in front of the commutator of the augmented Jordan algebra can be interpreted as picking out a forward time direction on the corresponding physical system. This can be made precise in the form of time orientations on the context category. (For more details we refer to \cite{FreDoe19b} and specifically \cite{AlfShu01,Doe14}). For our purposes the following notion will be sufficient.

\begin{definition}\label{def: time orientations}
	Let $\cH$ be a Hilbert space, $\BH$ the algebra of physical quantities, and $\mc{V}(\cH)$ the corresponding context category. The canonical time orientation on the context category is the map into the automorphisms on $\mc{V}(\cH)$, denoted by $\mathrm{Aut}(\mc{V}(\cH))$,
	\begin{align*}
	    \psi: \mathbb{R} \times \BHsa &\longrightarrow \mathrm{Aut}(\mc{V}(\cH)) \\
	    (t,a) &\longmapsto e^{ita}Ve^{-ita}\; .
	\end{align*}
	The context category $\VH$ together with the canonical time orientation $\psi$ on it is called the \emph{time-oriented context category} and denoted $\widetilde{\VH} = (\mc{V}(\cH),\psi)$.
\end{definition}

Note that by Thm.~\ref{Thm_WignerContextual}, every order automorphism on $\mc{V}(\cH)$ corresponds to conjugation by a unitary or anti-unitary operator. Since every anti-unitary is composed of the time-reversal operator and a unique unitary operator $e^{ita}$ (as in Def.~\ref{def: time orientations}), the former effectively causes a sign change in the parameter $t \in \mathbb{R}$, which is therefore naturally interpreted as the time parameter. Furthermore, it is straightforward to see that (infinitesimally) this corresponds to a sign change in the commutator (of the associative algebra). By the previous discussion, this corresponds exactly with the two different ways of extending a Jordan algebra of the form $\mc{J}(\BH)$ to a von Neumann algebra. We remark that for general Jordan algebras there are more ways to lift them to von Neumann algebras and thus also more possible time orientations on $\cN$.\footnote{There are corresponding notions of time-oriented presheaves over the context category as well (cf. \cite{Doe14}).} Succinctly,

\begin{center}
    \textbf{Time orientations encode the forward time direction in a quantum system.}
\end{center}

We also need the following definition of orientation-preserving global sections.

\begin{definition}
	Let $\cH_i$, $i=1,2$ be Hilbert spaces, $\mc{B}(\mc{H}_i)$ the algebras of physical quantities on either subsystem, and $\widetilde{\mc{V}(\cH_i)}$ the corresponding context categories with respective canonical time orientations $\psi_i$. A global section of the probabilistic presheaf $\gamma \in \Gamma(\PPi(\mc{V}_{1\& 2}))$ is called \emph{orientation-preserving with respect to $\psi = (\psi_1,\psi_2)$} if
	\begin{equation*}
	    \forall t \in \mathbb{R}, a \in \mc{B}(\mc{H}_1): \quad \Phi^\gamma \circ \psi_1(t,a) = \psi_2(t,\Phi^\gamma(a)) \circ \Phi^\gamma\; ,
	\end{equation*}
	where $\Phi^\gamma$ is the Jordan *-homomorphism in Lm~\ref{lm: global sections to Jordan homos}. The set of orientation-preserving global sections with respect to $\psi = (\psi_1,\psi_2)$ is denoted,
    \begin{equation*}
    	\Gamma(\PPi(\widetilde{\mc{V}_{1\& 2}})) := \{\gamma \in \Gamma(\PPi(\mc{V}_{1\& 2})) \mid \gamma \mathrm{\ is\ orientation-preserving\ with\ respect\ to\ } \psi\}\; ,
    \end{equation*}
    where $\widetilde{\mc{V}_{1\& 2}} = (\mc{V}_{1\& 2},\psi) = (\mc{V}(\cH_1),\psi_1) \times (\mc{V}(\cH_2),\psi_2) = \widetilde{\mc{V}(\cH_1)} \times \widetilde{\mc{V}(\cH_2)}$.
\end{definition}

Finally, we give our contextual reformulation of Bell's theorem.

\begin{theorem}\label{thm: Bell theorem in contextual form}
    (Bell's theorem in contextual form.)
    Let $\cH_i$, $i=1,2$ be Hilbert spaces $\mathrm{dim}(\cH_i) \geq 3$ finite, $\mc{B}(\mc{H}_i)$ the algebra of physical quantities, $\widetilde{\mc{V}(\cH_i)}$ the corresponding context categories with respective canonical time orientations $\psi_i$, and $\mc{S}(\mc{H}_i) \simeq \Gamma(\PPi(\mc{V}(\mc{B}(\mc{H}_i))))$ the corresponding state spaces. Then the state space of the composite system is given by
    \begin{equation*}
        \mc{S}(\cH_1 \otimes \cH_2) \simeq 
        \Gamma(\PPi(\widetilde{\mc{V}(\cH_1)} \times \widetilde{\mc{V}(\cH_2)}))\; .
    \end{equation*}
    Moreover, for $\mc{B}(\mc{H}_i)$ commutative with pure state spaces $\Sigma_i \simeq \Gamma(\Sig(\mc{V}(\mc{B}(\mc{H}_i))))$ one has,
    \begin{equation*}
        \Sigma_1 \times \Sigma_2 \simeq \Gamma(\Sig(\mc{V}(\cH_1) \times \mc{V}(\cH_2)))
    \end{equation*}
\end{theorem}


Thm.~\ref{thm: Bell theorem in contextual form} classifies state spaces of physical theories in terms of contextuality and context composition. First, it reproduces the known bound on classical correlations in terms of factorisability. As remarked in Sec.~\ref{sec: Correlation in Classical Theories}, the latter is a consequence of composition defined on state spaces, which for single-context theories is equivalent to context composition. What is more, Thm.~\ref{thm: Bell theorem in contextual form} rules out stronger, non-physical correlations such as PR-boxes, which are only (seemingly) allowed if few contexts are considered. In turn, this is a consequence of the order relations underlying composition defined via contexts in Eq.~(\ref{eq: product context category}), which implies  \emph{no-signalling}. Conversely, the latter relates contexts (and probability measures over them) as follows:
\begin{equation}\label{eq: no-signalling context order}
    (\tilde{V}_1,\tilde{V}_2) \subset_\mathrm{ns} (V_1,V_2) \ :\Lra \ (\tilde{V}_1 = V_1, \tilde{V}_2 \subset V_2) \ \mathrm{or} \ (\tilde{V}_1 \subset V_1, \tilde{V}_2 = V_2)
\end{equation}
It is straightforward to see that these conditions coincide with those in Eq.~(\ref{eq: product context category}) if we also demand transitivity. No-signalling together with the fact that each local subsystem possesses a quantum state space by Gleason's theorem (cf. `local quantumness' in \cite{Wehner2010}) therefore suffices to fix the correlations on the composite system, given by global sections of the Bell presheaf $\PPi(\mc{V}_{1\&2})$, to be quantum realisable. In fact, global sections $\Gamma(\PPi(\mc{V}_{1\&2}))$ \emph{exactly} correspond with quantum states in algebras with different time orientations. Since the surrounding arguments are of independent interest, we refer to our article \cite{FreDoe19b} for the proof of Thm. \ref{thm: Bell theorem in contextual form} and many more details.

Thm.~\ref{thm: Bell theorem in contextual form} is our reformulation of Bell's theorem in contextual form. Instead of only providing an upper bound on the amount of correlations that can exist in local hidden variable theories as, e.g. in the formulation in \cite{ClauserHorne1974}, it also shows that quantum correlations are singled out in a natural way by composition of systems via contexts rather than states. Note that this clearly implies that quantum correlations are bounded by the Tsirelson's bound \cite{Tsirelson1980}.


Finally, we note that Thm.~\ref{thm: Bell theorem in contextual form} provides a (first step towards a) well-defined notion of composition in the topos approach to quantum theory. All states of the composite system arise as global sections of the Bell presheaf (over the oriented composite context category). Hence, in order to describe the state space of the composite system, we do not need \emph{all} contexts of the composite system described in terms of the tensor product algebra, but merely product contexts, which arguably are the only ones operationally accessible. In turn, one may wonder whether knowing the state space of the composite system is sufficient to determine the poset of all contexts in the tensor product algebra. We will pursue this line of research elsewhere.

\section{Conclusion and Outlook}	\label{Sec_Conclusion}

In this article, we have shown that important structural theorems of quantum theory---Wigner's theorem, Gleason's theorem, the Kochen-Specker theorem, and Bell's theorem---fundamentally relate to contextuality. This might come as little surprise in the case of the Kochen-Specker theorem, yet other theorems had not been explicitly connected with contextuality before (to the best of our knowledge).

Wigner's theorem can be rephrased in terms of automorphisms on the partially ordered set of contexts, that is, maps preserving the context order, which are implemented by conjugation with unitary or anti-unitary operators. Hence, instead of demanding transition probabilities between pure states to be preserved, one can equivalently demand the order on contexts to be preserved instead.

The Kochen-Specker theorem is equivalent to the fact that the spectral presheaf has no global sections: given a `local' pure state in every context (each such state assigns sharp values to all observables in its context), there is no way of fitting these together in a consistent way. In other words, there are no dispersion-free quantum states. The `fitting together' here refers to the non-contextuality condition asserting that if an observable is contained in different contexts, then the value assigned to it by the different pure states must be the same.

Gleason’s theorem answers a similar local-to-global problem, yet instead of valuation functions, it considers measures on the physical quantities in a quantum system. From the perspective of presheaves, this is easily achieved by replacing pure states with mixed states \emph{locally}, that is, in every context, and by extending the restriction maps (from the spectral preheaf) to all probability measures, which thus become marginalisation constraints. Again, one asks for global sections, that is, probability assignments that are consistent \emph{globally}, or across contexts. In contrast to the spectral presheaf, global sections do exist in the case of the probabilistic presheaf, and by Gleason’s theorem bijectively correspond with quantum states (density matrices in finite dimensions). Gleason’s theorem therefore lifts quasi-linearity of probability measures in contexts to linearity on states.

Finally, Bell’s theorem attains a reformulation over contexts. Here, the crucial insight is the strong connection with composition of subsystems. In classical theories, composition is defined on the level of state spaces and thus in terms of the cartesian product. Bell’s original theorem can be read as a constraint on correlations between theories with states spaces composed in this way. However, from the perspective of physical contextuality, this is only justified for single-context systems. There, composition of states and composition of observables coincide. However, for multiple-context systems this is no longer the case. Shifting perspective from states to observables, one defines composition on the level of the context order instead. The corresponding Bell presheaf contains by far fewer contexts and thus also by far fewer constraints between local probability measures, which prima facie might allow for global sections that do not correspond to quantum states. Nevertheless, by combining several deep results including a generalised version of Gleason’s theorem, and adding the crucial notion of time orientation in local subsystems, this turns out not to be the case: all global sections of the Bell presheaf over the oriented product context category bijectively correspond with quantum states already. In other words, no-signalling and our consistency condition on time orientations in subsystems can be seen as the physical principles that rule out more general non-signalling correlations, thus replacing factorisability (stemming from the cartesian product construction on state spaces) in single-context systems.

Our contextual reformulation of Bell’s theorem—with composition defined on the level of contexts—thus unifies the classical and quantum case, it derives constraints for correlations of both: for the former we obtain the famous Bell inequalities, for the latter we obtain exactly the correlations realised in quantum theory, which rule out more general non-signalling correlations beyond the Tsirelson's bound.

In recent years, contextuality has been recognised more and more as a central feature of quantum theory \cite{IshBut98,Spekkens2005,Liangetal2011,AbramskyBrandenburger2011}. It has been argued that contextuality is a resource for quantum computation. In particular, for the quantum computing architecture known as measurement-based quantum computation (MBQC) it has been shown that contextuality allows to outperform certain non-contextual settings \cite{AndersBrowne2009,Raussendorf2013,deSilva2017,FrembsRobertsBartlett2018}. Moreover, in \cite{Howard2014} it is proven that contextuality is essential for magic state injection in the stabiliser formalism. Many connections to other resources exist, such as entanglement and negativity of the Wigner function \cite{Veitch2012,Delfosse2016}. For MBQC, a classification of contextuality in terms of group cohomology has been given in \cite{Raussendorf2016,BartlettRaussendorf2016}. Similar connections between contextuality and cohomology have also been obtained in \cite{Roumen2016} and within the sheaf theoretic formalism in \cite{AbramskyMansfieldBarbosa2011,BeerOsborne2018,Caru2016}. The latter framework provides a precise threshold for contextual computation by means of the contextual fraction \cite{AbramskyEtAl2015,AbramskyBarbosaMansfield2017,OkayTyhurstRaussendorf2018}.
These and similar results uncover the importance of contextuality for quantum computation, hinting at a universal classification and quantification of contextuality, which could pave the way for developments in future quantum computers.

Our work contributes to this research by substantially extending the scope of contextuality through the unified perspective it attains in the form of presheaves over the partial order of contexts. As we show, contextuality in this form underlies multiple and seemingly unrelated aspects in quantum theory.

\section*{Acknowledgments}

This work is supported through a studentship in the Centre for Doctoral Training on Controlled Quantum Dynamics at Imperial College funded by the EPSRC.

\noindent Note on research data policy: data sharing not applicable – no new data generated.



\section*{References}
\vspace{-1cm}
\bibliographystyle{acm}
\bibliography{bibliography}

\end{document}